\documentclass[preprints,article,accept,moreauthors,pdftex]{Definitions/mdpi}
\usepackage{amssymb}
\usepackage{amsmath}
\usepackage[ruled,vlined,linesnumbered]{algorithm2e}
\usepackage{multicol}
\usepackage{graphicx}
\usepackage{subcaption}
\usepackage{array,multirow}
\usepackage[export]{adjustbox}
\usepackage{setspace}
\newtheorem{lemma}[theorem]{Lemma}
\usepackage[labelformat=simple]{subcaption}

\DeclareCaptionLabelFormat{subcaptionlabel}{\normalfont(\textbf{#2}\normalfont)}
\setitemize{parsep=6pt,itemsep=0pt,leftmargin=*,labelsep=5.5mm}
\setenumerate{parsep=6pt,itemsep=0pt,leftmargin=*,labelsep=5.5mm}
\setlist[description]{itemsep=0mm}

\firstpage{1}
\pubvolume{xx}
\issuenum{1}
\articlenumber{5}
\pubyear{2020}
\copyrightyear{2020}
\history{Entropy \textbf{2020}, 22, 454; \url{doi.org/10.3390/e22040454}; Received: 04 March 2020; Accepted: 14 April 2020}

\Title{Image Encryption  Using Elliptic Curves and Rossby/Drift Wave Triads}

\Author{Ikram Ullah $^{1}$, Umar Hayat $^{1,}$* and Miguel D. Bustamante $^{2,}$* \orcidA{}}

\AuthorNames{Ikram Ullah, Umar Hayat and Miguel D. Bustamante}

\address{%
$^{1}$ \quad Department of Mathematics, Quaid-i-Azam University, Islamabad $45320$, Pakistan; ikram.ullah@math.qau.edu.pk~(I.U.)\\ 
$^{2}$ \quad School of Mathematics and Statistics, University College Dublin, Belfield, {Dublin 4}
, Ireland}

\corres{Correspondence: miguel.bustamante@ucd.ie~(M.D.B); umar.hayat@qau.edu.pk~(U.H.)}

\abstract{We propose an image encryption scheme based on {quasi-resonant Rossby/drift wave triads (related to elliptic surfaces)} and Mordell elliptic curves (MECs). By defining a total order on quasi-resonant triads, at a first stage we construct quasi-resonant triads using auxiliary parameters of {elliptic surfaces} in order to generate pseudo-random numbers. At a second stage, we employ an MEC to construct a dynamic substitution box (S-box) for the plain image. The generated pseudo-random numbers and S-box are used  to provide diffusion and confusion, respectively, in the tested image. We~test the proposed scheme against well-known attacks by encrypting {all gray images taken from the} USC-SIPI 
image database. Our experimental results indicate the high security of the newly developed scheme. Finally, via extensive comparisons we show that the new scheme outperforms other popular schemes.}


\begin{document}

\section{Introduction}
The exchange of confidential images via the internet is usual in today's life, even though the internet is an open source that is unsafe and unauthorized persons can steal useful or sensitive information. Therefore it is essential to be able to share images in a secure way. This goal is achieved by using cryptography. Traditional cryptographic techniques such as data encryption standard (DES) and advanced encryption standard (AES) are not suitable for image transmission because image pixels are usually highly correlated~\cite{Maqsood,signal45}. By contrast, DES and AES are ideal techniques for text encryption~\cite{signal52}, so researchers are trying to develop such techniques to meet the demand for reliable image delivery.

A number of image encryption schemes have been developed using different approaches~\cite{YangPan,Zhong, Li,Hua,Xie,NF,Luo, LiLi, Hua1,Wu,Yousaf}.{ Hua et al.~\cite{Hua1} developed a highly secure image encryption algorithm, where pixels are shuffled via the principle of the Josephus problem and diffusion is obtained by a filtering technology. \mbox{Wu et al.~\cite{Wu}} proposed a novel image encryption scheme by combining a random fractional discrete cosine transform (RFrDCT) and the chaos-based Game of Life (GoL). In their scheme, the desired level of confusion and diffusion is achieved by GoL and an XOR 
operation, respectively.} ``Confusion'' entails hiding the relation between input image, secret keys and the corresponding cipher image, and ``diffusion'' is an alteration of the value of each pixel in an input image~\cite{Maqsood}.

One of the dominant trends in encryption techniques is  chaos-based encryption~\cite{Ismail,Tang,Roa,Yu,Zhu,ElKamchouchi}. The~reason for this dominance is that the chaos-based encryption schemes are highly sensitive to the initial parameters. { However, there are certain chaotic cryptosystems that exhibit a lower security level due to the usage of chaotic maps with less complex behavior (see~\cite{h46}). This problem is addressed in~\cite{Hua2} by introducing a cosine-transform-based chaotic system (CTBCS) for encrypting images with higher security. Xu et al.~\cite{X14} suggested an image encryption technique based on fractional chaotic systems and verified experimentally the higher security of the underlying cryptosystem. \mbox{Ahmad et al.~\cite{A15}} highlighted certain defects of the above-mentioned cryptosytem by recovering the plain image without the secret key. Moreover, they proposed an enhanced scheme to thwart all kinds of attacks.}

The chaos-based algorithms also use pseudo-random numbers and substitution boxes (S-boxes) to create confusion and diffusion~\cite{Cheng,[ra23]}. Cheng et al.~\cite{Cheng} proposed an image encryption algorithm based on pseudo-random numbers and AES S-box. The pseudo-random numbers are generated using AES S-box and chaotic tent maps. The scheme is optimized by combining the permutation and diffusion phases, but the image is encrypted in rounds, which is time consuming. Belazi et al.~\cite{[ra23]} suggested an image encryption algorithm using a new chaotic map and logistic map. The new chaotic map is used to generate a sequence of pseudo-random numbers for masking phase. Then eight dynamic S-boxes are generated. The masked image is substituted in blocks via aforementioned S-boxes. The~substituted image is again masked by another pseudo-random sequence generated by the logistic map. Finally, the~encrypted image is obtained by permuting the masked image. The permutation is done by a sequence generated by the map function. This algorithm fulfills the security analysis but performs slowly due to the four cryptographic phases. In~\cite{signal46}, an image encryption method based on chaotic maps and dynamic S-boxes is proposed. The chaotic maps are used to generate the pseudo-random sequences and S-boxes. To break the correlation, pixels of an input image are permuted by the pseudo-random sequences. In a second phase the permuted image is decomposed into blocks. Then blocks are encrypted by the generated S-boxes to get the cipher image. From histogram analysis it follows that the suggested technique generates cipher images with a nonuniform distribution.

Similar to the chaotic maps, elliptic curves (ECs) are sensitive to input parameters, but EC-based cryptosystems are more secure than those of chaos~\cite{jia}. Toughi et al.~\cite{[ra]} developed a hybrid encryption algorithm using elliptic curve cryptography (ECC) and AES. The points of an EC are used to generate pseudo-random numbers and keys for encryption are acquired by applying AES to the pseudo-random numbers. The proposed algorithm gets the promising security but pseudo-random numbers are generated via the group law, which is time consuming. In~\cite{signal52}, a cyclic EC and a chaotic map are combined to design an encryption algorithm. The developed scheme overcomes the drawbacks of small key space but is unsafe to the known-plaintext/chosen-plaintext attack~\cite{Liu}. Similarly, \mbox{Hayat et al.~\cite{signal}} proposed an EC-based encryption technique. The stated scheme generates pseudo-random numbers and dynamic S-boxes in two phases, where the construction of S-box is not guaranteed for each input EC. Therefore, changing of ECs to generate an S-box is a time-consuming work. Furthermore, the~generation of ECs for each input image makes it insufficient.

Based on the above discussion, we propose an improved image encryption algorithm, based on quasi-resonant Rossby/drift wave triads \cite{Busta,UmarHayat} (triads, for short) and Mordell elliptic curves (MECs). The triads are utilized in the generation of pseudo-random numbers and MECs are employed to create dynamic S-boxes.
The proposed scheme is novel in that it introduces the technique of pseudo-random numbers generation using triads, which is faster than generating pseudo-random numbers by ECs.
 Moreover, the scheme does not require to separately generate triads for each input image of the same size. In the present scheme, MECs are used opposite to~\cite{signal}, in the sense that now, for each input image, the generation of a dynamic S-box is guaranteed~\cite{Iso}. Finally, extensive performance analyses and comparisons reveal the efficiency of the proposed scheme.

\textls[-15]{This paper is organized as follows. Preliminaries are described in Section~\ref{pre}. In Section~\ref{scheme}, the~proposed encryption algorithm is explained in detail. Section~\ref{SEAnalysis} provides the experimental results as well as a comparison between the proposed method and other existing popular schemes. Lastly, conclusions are presented in Section~\ref{Conclusion}.}

\section{Preliminaries}\label{pre}
{\bf Barotropic vorticity equation:} The barotropic vorticity equation (in the so-called $\beta$-plane approximation) {is one of the simplest two-dimensional models of the large-scale dynamics of a shallow layer of fluid on the surface of a rotating sphere. It is described in mathematical terms by the partial differential equation \begin{equation}\label{BVE}
\frac{\partial}{\partial t} (\nabla^2\psi - F\psi)+\Big(\frac{\partial\psi}{\partial x}\frac{\partial\nabla^2\psi}{\partial y}-\frac{\partial\psi}{\partial y}\frac{\partial\nabla^2\psi}{\partial x}\Big)+\gamma\frac{\partial\psi}{\partial x}=0,
\end{equation}
where $\psi(x,y,t)\in \mathbb{R}$ represents the geopotential height, $\gamma$ is the Coriolis parameter, a real constant measuring the variation of the Coriolis force with latitude ($x$ represents longitude and $y$ represents latitude) and $F$ is a non-negative real constant representing the inverse of the square of the deformation radius. We assume periodic boundary conditions: $\psi(x+2\pi,y,t) = \psi(x,y+2\pi,t) =\psi(x,y,t)$ for all $x,y,t \in \mathbb{R}$. In the literature Equation~(\ref{BVE}) is also known as the Charney–Hasegawa–Mima equation (CHM) \cite{charney1948scale,hasegawa1978pseudo,connaughton2010modulational,harris2013percolation, galperin_read_2019}. This equation accepts harmonic solutions, known as Rossby waves, which are solutions of both the linearized form and the whole (nonlinear) form of Equation~(\ref{BVE}). A Rossby wave solution is given explicitly by the parameterized function $\psi_{(k,l)}(x,y,t)=\Re\{A\,{\mathrm e}^{i(kx+ly-\omega(k,l)t)}\}$, where $A \in \mathbb{C}$ is an arbitrary constant, $\omega(k,l)=-\frac{\gamma\, k}{k^2+l^2 + F}$ is the so-called dispersion relation, and $(k,l)\in\mathbb{Z}^2$ is called the wave vector. For simplicity, we take $\gamma=-1$ and $F=0$ in what follows \cite{Busta,UmarHayat}.

\textbf{Resonant triads:} As Equation (\ref{BVE}) is nonlinear, modes with different wave vectors tend to couple and exchange energy. If the nonlinearity is weak, this exchange happens to be quite slow and is more efficient amongst groups of modes that are in \emph{resonance}.  To the lowest order of nonlinearity in Equation~(\ref{BVE}), approximate solutions known as resonant triad solutions can be constructed via linear combinations of the form
$$\psi(x,y,t) = \Re\{A_1\,{\mathrm e}^{i(k_1x+l_1y-\omega(k_1,l_1)t)}+A_2\,{\mathrm e}^{i(k_2 x+l_2y-\omega(k_2,l_2)t)}+A_3\,{\mathrm e}^{i(k_3 x+l_3y-\omega(k_3,l_3)t)}\}\,,$$
where $A_1, A_2, A_3$ are slow functions of time (they satisfy a closed system of ODEs, 
not shown here), and the wave vectors $(k_1,l_1),(k_2,l_2)$ and $(k_3,l_3)$ satisfy the Diophantine system of equations:
\begin{equation}\label{eqns}
 k_1+k_2=k_3, \,\,\,\,\,\, l_1+l_2=l_3 \,\,\,\,\,\,{\rm and}\,\,\,\,\,\, \omega_1+\omega_2=\omega_3,
\end{equation}
for $\omega_i=\omega(k_i,l_i),i=1,2,3$. A set of three wavevectors satisfying Equations (\ref{eqns}) is called a resonant triad. Solutions can be found analytically via a rational transformation to elliptic surfaces (see below).

\textbf{Quasi-resonant triads and detuning level:} If, in (\ref{eqns}), the equation $\omega_1+\omega_2=\omega_3$ is replaced by the inequality $|\omega_1+\omega_2-\omega_3| \leq \delta^{-1}$, for a large positive number $\delta$, then the triad becomes a quasi-resonant triad and $\delta^{-1}$ is known as the detuning level of the quasi-resonant triad. It is possible to construct quasi-resonant triads via downscaling of resonant triads that have very large wave vectors~\cite{Busta}. For~simplicity, in what follows we simply call a quasi-resonant triad  a triad and denote it by $\Delta$. Finally, to avoid over-counting of triads we will impose the condition $k_3 > 0$.}

{\bf Rational transformation:} In~\cite{Busta}, wave vectors are explicitly expressed in terms of rational variables $X,Y$ and $D$  as follows:
\begin{equation}\label{inv}
 \frac{k_1}{k_3}=\frac{X}{Y^2+D^2}, \,\,\,\,\,\,
 \frac{l_1}{k_3}=\Big(\frac{X}{Y}\Big)\Big(1-\frac{D}{Y^2+D^2}\Big),\,\,\,\,\,\,
 \frac{l_3}{k_3}=\frac{D-1}{Y}.
 \end{equation}
In the case $F=0$, the rational variables $X,Y, D$ lie on an elliptic surface. 
The transformation is bijective and its inverse mapping is given by:
 \begin{equation}\label{map}
 X=\frac{k_{3}(k_{1}^2+l_{1}^2)}{k_{1}(k_{3}^2+l_{3}^2)},\,\,\,\,\,\,
 Y=\frac{k_{3}(k_{3}l_{1}-k_{1}l_{3})}{k_{1}(k_{3}^2+l_{3}^2)},\,\,\,\,\,\,
 D=\frac{k_{3}(k_{3}k_{1}-l_{1}l_{3})}{k_{1}(k_{3}^2+l_{3}^2)}.
 \end{equation}

\textbf{New parameterization:} In~\cite{Kopp}, Kopp parameterized the resonant triads and in terms of parameters $u$ and $t$ it follows by \textnormal{\cite{Kopp}} (Equation~(1.22)) that:
\begin{align}
\frac{k_{1}}{k_{3}}=& (t^2+u^2)(t^2-2u+u^2)/(1-2u),\\
\frac{l_{3}}{k_{3}}=& \big(u(2u-1)+(t^2+u^2)(t^2-2u+u^2)\big)/\big(t(1-2u)\big),\\
\frac{l_{1}}{k_{3}}=& (t^2+u^2)\big((2u-1)+u(t^2-2u+u^2)\big)/\big(t(1-2u)\big).
\end{align}
In $2019$, Hayat et al.~\cite{{UmarHayat}} found a new parameterization of $X,Y$ and $D$ in terms of auxiliary parameters $a,b$ and hence $\frac{k_{1}}{k_{3}},\frac{l_{3}}{k_{3}}$ and $\frac{l_{1}}{k_{3}}$ are given by:
\begin{align}
\frac{k_{1}}{k_{3}}=&\frac{\big(a^{2}+b(2-3b)+1\big)^{3}}{(a^{2}-3b^{2}-2b+1)\big(2(11-3a^{2})b^{2}+(a^{2}+1)^{2}-16ab+9b^{4}\big)},\label{pa1}\\
\frac{l_{3}}{k_{3}}=&\frac{6(a^{2}+a-1)b^{2}-(a+1)^{2}(a^{2}+1)+4ab-9b^{4}}{(a^{2}-3b^{2}-1)(a^{2}-3b^{2}-2b+1)},\label{pa2}\\
\frac{l_{1}}{k_{3}}=&\frac{\big(a^{2}+b(2-3b)+1\big)}{\footnotesize\begin{matrix}
(a^{2}-3b^{2}-1)(a^{2}-3b^{2}-2b+1)\big(2(11-2a^{2})b^{2}+(a^{2}+1)^{2}-16ab+9b^{4}\big)&\\
\times[a^{6}+2a^{5}+a^{4}(-9b^{2}-6b+3)-4a^{3}(3b^{2}+2b-1)+3a^{2}(3b^{2}+2b-1)^{2}&\\
+2a(9b^{4}+12b^{3}+14b^{2}-4b+1)-(3b^{2}+1)^{2}(3b^{2}+6b-1)]\end{matrix}}\label{pa3}.
\end{align}

\textbf{Elliptic curve (EC): } Let $\mathbb{F}_{p}$ be a finite field for any prime $p$, then an EC $E_{p}$ over $\mathbb{F}_{p}$ is defined by
\begin{equation}\label{EC}
y^{2}\equiv x^{3}+bx+c \pmod{p},
\end{equation}
where $b,c\in \mathbb{F}_{p}$. The integers $b,c$ and $p$ are called parameters of an EC. The number of all $(x,y)\in \mathbb{F}_{p}^{2}$ satisfying the congruence~(\ref{EC}) is denoted by $\#E_{p}$.

\textbf{Mordell elliptic curve (MEC): } In the special but important case $b=0$, the above EC is known as an MEC and is represented by
\begin{equation}\label{MMEC}
y^{2}\equiv x^{3}+c \pmod{p}.
\end{equation}
For $p\equiv2 \pmod{3}$, there are exactly $p+1$ points $(x,y)\in \mathbb{F}_{p}^{2}$ satisfying the congruence~(\ref{MMEC}), see \cite{Wash} for further details.

If points on $E_{p}$ are ordered according to some total order $\prec$ then $E_{p}$ is said to be an ordered EC. Recall that total order is a binary relation which possesses the reflexive, antisymmetric and transitive properties. Azam et al.~\cite{Azam1} introduced a total order known as a natural ordering on MECs given by\\
\[
(x_{1},y_{1})\prec (x_{2},y_{2}) \Leftrightarrow
\begin{cases}
{\rm either \ } x_{1}<x_{2}, {\rm \ or} \\[5pt]
x_{1}=x_{2} {\rm \ and \ } y_{1}<y_{2}, %
\end{cases}
\]
and generated efficient S-boxes using the aforesaid ordering. We will use natural ordering to generate S-boxes. Thus from here on $E_{p}$ stands for a naturally ordered MEC unless it is specified otherwise.

\section{The Proposed Encryption Scheme}\label{scheme}
The proposed encryption scheme is based on pseudo-random numbers and S-boxes. The~pseudo-random numbers are generated using quasi-resonant triads.
To get an appropriate level of diffusion we need to properly order the $\Delta$s. For this purpose we define a binary relation $\lesssim$ as follows.\\
\subsection{Ordering on Quasi-Resonant Triads}
Let $\Delta,\Delta'$ represent the triads $(k_{i},l_{i}),(k_{i}',l_{i}'),i=1,2,3$, respectively, then
\[
\Delta \lesssim\Delta' \Leftrightarrow
\begin{cases}
{\rm either \ } a < a', {\rm \ or} \\[5pt]
a=a' {\rm \ and \ } b < b', {\rm \ or}\\[5pt]
a=a', b = b' {\rm \ and \ } k_{3}\leq k_{3}',
\end{cases}
\]\\
where $a,b$ and $a',b'$ are the corresponding auxiliary parameters of $\Delta$ and $\Delta'$, respectively.
\begin{lemma}
If $T$ denotes the set of $\Delta$s in a box of size $L$, then $\lesssim$ is a total order on $T$.
\end{lemma}
\begin{proof}
The reflexivity of $\lesssim$ follows from $a=a, b = b$ and $k_{3}= k_{3}$ and hence $\Delta \lesssim\Delta.$ As for antisymmetry we suppose $\Delta \lesssim\Delta'$ and $\Delta' \lesssim\Delta$. Then, by definition $a \leq a'$ and $a' \leq a$, which imply $a=a'$. Thus we are left with two results: $b \leq b'$ and $b' \leq b$, which imply $b =b'$. Thus, we obtain the results $k_{3}\leq k_{3}'$ and $k_{3}'\leq k_{3}$, which ultimately give $k_{3} = k_{3}'$. Solving Equations~(\ref{pa1})--(\ref{pa3}) for the obtained values, we get $k_{1}= k_{1}',l_{3}= l_{3}'$ and from Equation~(\ref{eqns}) it follows that $l_{2}= l_{2}'$. Consequently $\Delta =\Delta'$ and $\lesssim$ is antisymmetric. As for transitivity, let us assume $\Delta \lesssim\Delta'$ and $\Delta' \lesssim\Delta''$. Then $a\leq a'$ and $a'\leq a''$, implying  $a\leq a''$. If $a < a''$, then transitivity follows. If $a = a''$, then $a'=a''$ too. Thus, $b\leq b'$ and $b' \leq b''$, so $b\leq b''$. If $b<b''$, then transitivity follows. If $b=b''$, then $b'=b''$ too. Thus, $k_{3}\leq k_{3}'$ and $k_{3}'\leq k_{3}''$, implying $k_{3}\leq k_{3}''$ and hence transitivity follows: $\Delta \lesssim\Delta''$.
\end{proof}
Let $\stackrel{*}{T}$ stand for the set of $\Delta$s ordered with respect to the order $\lesssim$. The main steps of the proposed scheme are explained as follows.
\subsection{Encryption}
\noindent \textbf{A. Public parameters:} In order to exchange the useful information the sender and receiver should agree on the public parameters described as below:
\begin{enumerate}[align=parleft,leftmargin=*,labelsep=7.5mm]
	\item[(1)] Three sets: choose three sets $\mathcal{A}_{i}=[A_{i},B_{i}],i=1,2,3$ of consecutive numbers with unknown step sizes, where the end points $A_{i},B_{i},i=1,2,3$ are rational numbers.
\item[(2)] A total order: select a total order $\prec$ so that the triads generated by the above-mentioned sets may be arranged with respect to that order.
\end{enumerate}

Suppose that $P$ represents an image of size $m\times n$ to be encrypted, and the pixels of $P$ are arranged in column-wise linear ordering. Thus, for positive integer $i\leq mn$, $P(i)$  represents the $i$-th pixel value in linear ordering. {Define $S_{\rm P}$ as the sum of all pixel values of the image $P$.} Then the proposed scheme chooses the secret keys in the following ways.\\

\noindent \textbf{B. Secret keys:} To generate confusion and diffusion in an image, the sender chooses the secret keys as~follows.
\begin{enumerate}[align=parleft,leftmargin=*,labelsep=7.5mm]
\item[(1)] Step size: select positive integers $a_{i},b_{i}$ to construct the step sizes $\alpha_{i}=\frac{a_{i}}{b_{i}}$ of $\mathcal{A}_{i},i=1,2$. Additionally, choose a non-negative integer $a_{3}$ as a step size of $\mathcal{A}_{3}$ in such a way that $\prod_{i=1}^{3}n_{i}\geq mn$, where $\#\mathcal{A}_{i}=n_{i}$ represents the number of elements in $\mathcal{A}_{i}$.
\item[(2)] Detuning level: fix some posive integer $\delta$ to find the detuning level $\delta^{-1}$ allowed for the triads.
\item[(3)] Bound: select a positive integer $L$ such that $|k_{i}|, |l_{i}|\leq L$ for $i=1,2,3.$ This condition is imposed in order to bound the components of the triad wave vectors. Furthermore, choose an integer $t$ to find $r=\lfloor S_{\rm P}/t\rceil$, where $\lfloor\cdot\rceil$ gives the nearest integer when $S_{\rm P}$ is divided by $t$. The reason for choosing such a $t$ is to generate key-dependent S-boxes and the integer $r$ is used to diffuse the components of triads.
\item[(4)] A prime: select a prime $p\geq 257$ such that $p\equiv2 \pmod{3}$ as a secret key for computing nonzero  $c\equiv S_{\rm P}+t\pmod{p}$  to generate an S-box $\zeta_{E_{p}}(p,t,S_{\rm P})$ on the $E_{p}$. { The S-box construction technique is made clear in Algorithm~\ref{SAL}, and the S-box generated for $p=1607, t=182$ and $S=0$ by Algorithm~\ref{SAL} is shown in Table~\ref{SbT}. Furthermore, the cryptographic properties of the said S-box are evaluated in Sections~\ref{Sb1} and~\ref{Sb2}. }
\end{enumerate}
\begin{algorithm}[H]
\DontPrintSemicolon
\tcc{$B$ is a set of points $(x,y)$ satisfying $E_{p}$, $B(i)$ is $i$-th point of $B$ and $y_{i}$ stands for $y$-component of point $B(i)$.}
\SetKwInOut{Input}{Input}\SetKwInOut{Output}{Output}
\Input{A prime $p\equiv2 \pmod{3}$ and two integers $t$ and $S$ such that {$c=S+t$} and $S+t\not\equiv 0 \pmod{p}$.}
\Output{An S-box $\zeta_{E_{p}}(p,t,S)$. }
$B:=\emptyset$;\;
$Y:=[0, (p-1)/2]$;\;
$i\leftarrow 0$;   \;
\For{$x\in [0,p-1]$}{
    \For{$y\in Y$}{
    \If{$y^2\equiv x^3+c \pmod{p}$}{
    $i\leftarrow i+1$; $B(i):=(x,y)$;\\
    \If{$y\not=0$}{
    $i\leftarrow i+1$; $B(i):=(x,p-y)$;\\
    break;\\
    }
    }

 }
 {$Y=Y-\{y\}$};
}
$\zeta_{E_{p}}(p,t,S)=\{y_{i}\in B(i): 0\leq y_{i}< 256\}.$

\caption{{Construction of $8\times 8$ S-box.}}
\label{SAL}
\end{algorithm}
\begin{table}[H]
\caption{{The obtained S-box $\zeta_{E_{1607}}(1607,182,0)$}.}
\label{SbT}
\centering
\resizebox{\columnwidth}{!}
{
\bgroup
\begin{tabular}{cccccccccccccccc}
\toprule
220 & 118 & 17  & 158 & 25  & 138 & 33  & 196 & 247 & 252 & 15  & 226 & 135 & 177 & 232 & 83  \\
161 & 70  & 107 & 186 & 137 & 236 & 21  & 142 & 131 & 103 & 54  & 58  & 217 & 181 & 201 & 172 \\
91  & 84  & 223 & 89  & 29  & 156 & 136 & 14  & 69  & 99  & 164 & 171 & 35  & 188 & 76  & 139 \\
153 & 16  & 198 & 227 & 32  & 10  & 115 & 122 & 184 & 61  & 208 & 225 & 213 & 106 & 94  & 56  \\
165 & 40  & 245 & 189 & 163 & 239 & 193 & 194 & 129 & 175 & 241 & 141 & 130 & 231 & 215 & 127 \\
151 & 199 & 105 & 22  & 148 & 39  & 179 & 173 & 78  & 248 & 81  & 23  & 75  & 55  & 146 & 109 \\
195 & 251 & 178 & 170 & 162 & 206 & 228 & 169 & 147 & 28  & 210 & 221 & 80  & 121 & 202 & 77  \\
9   & 74  & 197 & 31  & 26  & 154 & 145 & 44  & 47  & 82  & 43  & 60  & 117 & 250 & 88  & 191 \\
67  & 8   & 174 & 93  & 1   & 20  & 128 & 53  & 218 & 237 & 96  & 72  & 3   & 65  & 6   & 253 \\
150 & 101 & 119 & 87  & 160 & 133 & 108 & 57  & 41  & 64  & 51  & 49  & 185 & 243 & 2   & 249 \\
167 & 50  & 205 & 183 & 97  & 114 & 48  & 27  & 246 & 254 & 124 & 92  & 19  & 134 & 159 & 95  \\
24  & 224 & 111 & 62  & 116 & 168 & 200 & 86  & 79  & 143 & 126 & 112 & 45  & 71  & 125 & 13  \\
5   & 216 & 187 & 222 & 7   & 113 & 238 & 36  & 204 & 52  & 140 & 46  & 240 & 85  & 207 & 4   \\
152 & 104 & 235 & 190 & 242 & 68  & 63  & 203 & 230 & 176 & 180 & 59  & 157 & 244 & 66  & 212 \\
34  & 90  & 120 & 0   & 30  & 166 & 37  & 255 & 38  & 110 & 211 & 233 & 11  & 155 & 209 & 219 \\
192 & 12  & 144 & 73  & 182 & 132 & 98  & 214 & 42  & 102 & 18  & 149 & 123 & 229 & 100 & 234\\ \bottomrule
\end{tabular}
\egroup
}
\end{table}

The positive integers $a_{1}, b_{1}, a_{2}, b_{2}, a_{3},\delta,L,S_{\rm P},t$ and $p$  are secret keys. Here it is mentioned that the parameters $a_{1}, b_{1}, a_{2}, b_{2}, a_{3},\delta$ and $L$ are used to generate $mn$ triads in a box of size $L$. The generation of triads is explained step by step in Algorithm~\ref{triad}. These triads along with keys $S_{\rm P}$ and $t$ are used to generate the sequence $\beta_{\stackrel{*}{T}}(t,S_{\rm P})$ of pseudo-random numbers.

\begin{algorithm}[H]
\DontPrintSemicolon
\tcc{T is a set containing the Quasi-resonant triads, while $m$ and $n$ are the dimensions of an input image.}
\SetKwInOut{Input}{Input}\SetKwInOut{Output}{Output}
\Input{Three sets $\mathcal{A}_{i}, i=1,2,3$, inverse detuning level $\delta$, bound $L$, two positive integers $m$ and $n$.}
\Output{Quasi-resonant triads}
$T:=\emptyset$;\;
$c_{1}\leftarrow 0,c_{2}\leftarrow 1$ ;   \;
\For{$a\in\mathcal{A}_{1}$}{
    \For{$b\in\mathcal{A}_{2}$}{
    $c_{1}\leftarrow c_{1}+1$;\\
    Calculate and store the values of $k_{1}^{\prime}(c_{1}), l_{3}^{\prime}(c_{1})$, and
    $l_{1}^{\prime}(c_{1})$ for each pair $(a,b)$ using Equations~(\ref{pa1})--(\ref{pa3}).

 }
}
\For{$c_{2}\in [1, c_{1}]$}{
\For{$k_{3}\in \mathcal{A}_{3}$}{
$k_{1}=\lfloor(k_{1}^{\prime}(c_{2})*k_{3})\rceil, l_{3}=\lfloor(l_{3}^{\prime}(c_{2})*k_{3})\rceil$ and $l_{1}=\lfloor(l_{1}^{\prime}(c_{2})*k_{3})\rceil$;\\
$k_{2}=k_{3}-k_{1}, l_{2}=l_{3}-l_{1}$ and $\omega_{i}=k_{i}/(k_{i}^2+l_{i}^2),i=1,2,3;$\\
$\omega_{4}=\omega_{3}-\omega_{2}-\omega_{1};$\\
\If{$|\omega_{4}|<\delta^{-1}$  and $0<|k_{i}|,|l_{i}|<L,i=1,2,3$}{
$T:=T\cup\{\Delta\}$;\\

        }
\If {\#T=mn}{break;}
}
break;
}

{Sort $T$ with respect to the ordering $\lesssim$ to get $\stackrel{*}{T}$.}
\caption{Generating quasi-resonant triads.}
\label{triad}
\end{algorithm}
Thus $\Delta_{\rm j}$ represents the $j$-th triad in ordered set $\stackrel{*}{T}$. Moreover, $(k_{ji},l_{ji}),i=1,2,3$ are the components of $\Delta_{\rm j}$ . In Algorithm~\ref{RN}, the generation of $\beta_{\stackrel{*}{T}}(t,S_{\rm P})$ is interpreted.\\
\begin{algorithm}[H]\setstretch{1.50}
\DontPrintSemicolon
\SetAlgoLined
\SetKwInOut{Input}{Input}\SetKwInOut{Output}{Output}
\Input{An ordered set $\stackrel{*}{T}$, an integer $t$ and a plain image $P$.}
\Output{Random numbers sequence $\beta_{\stackrel{*}{T}}(t,S_{\rm P})$.}
$Tr(j):= |{rk_{j1}}|+|{l_{j1}}|+|{k_{j2}}|$;   \;
$\beta_{\stackrel{*}{T}}(t,S_{\rm P})(j)=(Tr(j)+S_{\rm P})\pmod{256}$;\;


\caption{Generating the proposed pseudo-random sequence.}
\label{RN}
\end{algorithm}\vspace{12pt}

The proposed sequence $\beta_{\stackrel{*}{T}}(t,S_{\rm P})$ is cryptographically a good source of pseudo-randomness because triads are highly sensitive to the auxiliary parameters $(a,b)$~\cite{UmarHayat} and inverse detuning level $\delta$. It is shown in~\cite{Busta} that the intricate structure of clusters formed by triads depends on the chosen $\delta$, and the size of the clusters increases as the inverse detuning level increases. Moreover, the generation of triads is rapid due to the absence of modular operation.\\

\noindent\textbf{C. Performing diffusion.} To change the statistical properties of an input image, a diffusion process is performed. While performing the diffusion, the pixel values are changed using the sequence $\beta_{\stackrel{*}{T}}(t,S_{\rm P})$. Let $M_{\rm P}$ denote the diffused image for a plain image $P$. The proposed scheme alters the pixels of $P$ according to:
\begin{equation}\label{defus}
M_{\rm P}(i)=\beta_{\stackrel{*}{T}}(t,S_{\rm P})(i)+P(i)\pmod{256}.
\end{equation}
\textbf{D. Performing confusion.} A nonlinear function causes confusion in a cryptosystem, and nonlinear components are necessary for a secure data encryption scheme. The current scheme uses the dynamic S-boxes to produce the confusion in an encrypted image. If $C_{\rm P}$ stands for the encrypted image of $P$, then confusion is performed as follows:
\begin{equation}\label{confus}
C_{\rm P}(i)=\zeta_{E_{p}}(p,t,S_{\rm P})(M_{\rm P}(i)).
\end{equation}
\begin{lemma}
If $\#\mathcal{A}_{i}=n_{i},i=1,2,3$ and $p$ is a prime chosen for the generation of an S-box, then the time complexity of the proposed encryption scheme is max$\{\mathcal{O}(n_{1}n_{2}n_{3}), p^2\}$.
\end{lemma}
\begin{proof}
The computation of all possible values of $k_{1}',l_{3}'$ and $l_{1}'$ in Algorithm~\ref{triad} takes $\mathcal{O}(n_{1}n_{2})$ time. Similarly the time complexity for generating $\stackrel{*}{T}$ is $\mathcal{O}(c_{1}n_{3})$ but $c_{1}$ executes $n_{1}n_{2}$ times. Thus the time required by $\stackrel{*}{T}$ and hence by  $\beta_{\stackrel{*}{T}}(t,S_{\rm P})$ is $\mathcal{O}(n_{1}n_{2}n_{3})$. Additionally, Algorithm~\ref{SAL} shows that the proposed S-box can be constructed in $\mathcal{O}(p^2)$ time. Thus the time complexity of the proposed scheme is max$\{\mathcal{O}(n_{1}n_{2}n_{3}), p^2\}$.\\
\end{proof}

\noindent \textbf{Example 1.}
\emph{{In order to have a clear picture of the proposed cryptosystem, we explain the whole procedure using the following hypothetical $4\times 4$ image. For example, let $I$ represent the plain image of Lena$_{256\times256}$, and let $P$ be the subimage of $I$ consisting of the intersection of the first four rows and the first four columns of $I$ as shown in Table~\ref{P44}, whereas the column-wise linearly ordered image $P$ is shown in Table~\ref{LOO}.}}
\begin{table}[H]
\caption{{Plain image $P$.}}
\label{P44}%
\centering
\begin{tabular}{cccc}
\toprule
 162 & 162 & 162 & 163 \\
162 & 162 & 162 & 163 \\
162 & 162 & 162 & 163 \\
160 & 163 & 160 & 159\\ \bottomrule
\end{tabular}
\end{table}\vspace{-18pt}
\begin{table}[H]
\caption{{Linear ordering of image $P$.}}
\label{LOO}%
\centering
\begin{tabular}{cccc}
\toprule
$P(1)$ & $P(5)$ & $P(9)$ & $P(13)$ \\
$P(2)$ & $P(6)$ & $P(10)$ & $P(14)$ \\
$P(3)$ & $P(7)$ & $P(11)$ & $P(15)$ \\
$P(4)$ & $P(8)$ & $P(12)$ & $P(16)$\\ \bottomrule
\end{tabular}
\end{table}
\noindent {We have $S_{\rm P}=2589$ and $c=247$ and the values of other parameters are described in Section~\ref{Analysis}. The~corresponding $16$ triads are obtained by Algorithm~\ref{triad} as shown in Table~\ref{T16}.}
\begin{table}[H]
\caption{{The corresponding set $\stackrel{*}{T}$ for image $P$.}}
\label{T16}
\centering
\begin{tabular}{cccccccccccccc}
\toprule
\boldmath{$\Delta_{j}$}& \boldmath{$k_{1}$} & \boldmath{$l_{1}$} & \boldmath{$k_{2}$}& \boldmath{$l_{2}$} & \boldmath{$k_{3}$} &  \boldmath{$l_{3}$}  & \boldmath{$\Delta_{j}$} & \boldmath{$k_{1}$} & \boldmath{$l_{1}$} & \boldmath{$k_{2}$}& \boldmath{$l_{2}$} & \boldmath{$k_{3}$ }&  \boldmath{$l_{3}$}     \\
\midrule
$\Delta_{1}$ & $-$1128 & 1152 & 1529 & 668 & 401 & 1820 & $\Delta_{9}$ & $-$1240 & 1267 & 1681 & 735 & 441 & 2002 \\
$\Delta_{2}$ & $-$1142 & 1167 & 1548 & 676 & 406 & 1843 & $\Delta_{10}$ & $-$1254 & 1282 & 1700 & 743 & 446 & 2025 \\
$\Delta_{3}$ & $-$1156 & 1181 & 1567 & 685 & 411 & 1866 & $\Delta_{11}$ & $-$1268 & 1296 & 1719 & 751 & 451 & 2047 \\
$ \Delta_{4}$& $-$1170 & 1195 & 1586 & 694 & 416 & 1889 & $\Delta_{12}$ & $-$1282 & 1310 & 1738 & 760 & 456 & 2070 \\
$\Delta_{5}$ & $-$1184 & 1210 & 1605 & 701 & 421 & 1911 & $\Delta_{13}$ & $-$1296 & 1325 & 1757 & 768 & 461 & 2093 \\
$\Delta_{6}$ & $-$1198 & 1224 & 1624 & 710 & 426 & 1934 & $\Delta_{14}$ & $-$1310 & 1339 & 1776 & 776 & 466 & 2115 \\
$\Delta_{7}$ & $-$1212 & 1238 & 1643 & 719 & 431 & 1957 & $\Delta_{15}$ & $-$1325 & 1353 & 1796 & 785 & 471 & 2138 \\
$\Delta_{8}$ & $-$1226 & 1253 & 1662 & 726 & 436 & 1979 & $\Delta_{16}$ & $-$1339 & 1368 & 1815 & 793 & 476 & 2161\\
\bottomrule
\end{tabular}
\end{table}
From $S_{\rm P}=2589$ and $t=2$, it follows that $r=1295$ and hence by application of Algorithm~\ref{RN} the terms of $\beta_{\stackrel{*}{T}}(2,2589)$  are listed in Table~\ref{seq}. Moreover, the S-box $\zeta_{E_{293}}(293,2,2589)$ is constructed by Algorithm~\ref{SAL}, giving the mapping $\zeta_{E_{293}}(293,2,2589):\{0,1,\ldots, 255\}\rightarrow \{0,1,\ldots, 255\}$, which maps the list $(0,\ldots, 255)$ to the list $(80, 213, 29, 113, 180, 2, 119, 174, 10, 103, 190, 120, 173, 99, 194, 126,167, 42, 251, 78, 215, 84, 209, 93, 200, 130,$
$163, 32, 17, 117, 176, 62, 231, 110, 183, 56, 237, 75, 218, 127, 166, 73, 220, 13, 91, 202, 28, 129, 164, 118, 175, 69,$ $224, 50, 243, 100, 193, 137, 156, 89, 204, 12, 63, 230, 74, 219, 4, 131, 162, 134, 159, 123, 170, 90, 203, 70, 223, 87,$ $206, 59, 234, 145, 148, 58, 235, 57, 236, 65, 228, 15, 112, 181, 52, 241, 76, 217, 60, 233, 121, 172, 68, 225, 51, 242,$ $135, 158, 41, 252, 21, 142, 151, 26, 25, 40, 253, 96, 197, 136, 157, 9, 116, 177, 122, 171, 45, 248, 115, 178, 102, 191,$  $67, 226, 95, 198, 143, 150, 133, 160, 98, 195, 3, 94, 199, 30, 104, 189, 132, 161, 8, 64, 229, 144, 149, 140, 153, 14,$ $85, 208, 20, 6, 109, 184, 125, 168, 92, 201, 19, 53, 240, 31, 66, 227, 35, 82, 211, 108, 185, 139, 154, 33, 16, 86, 207,$ $128, 165, 5, 71, 222, 38, 255, 23, 0, 81, 212, 1, 141, 152, 111, 182, 138, 155, 49, 244, 22, 106, 187, 105, 188, 36, 54,$ $239, 46, 247, 43, 250, 97, 196, 27, 11, 24, 44, 249, 83, 210, 61, 232, 39, 254, 7, 72, 221, 77, 216, 47, 246, 107, 186,$ $48, 245, 55, 238, 124 169, 34, 79, 214, 88, 205, 114, 179, 37, 18, 146, 147, 101, 192)$. 
\begin{table}[H]
\caption{{Pseudo-random sequence for plain image $P.$}}
\label{seq}%
\centering
\begin{tabular}{cccc}
\toprule
$\beta_{\stackrel{*}{T}}(2,2589)(1)=188$ & $\beta_{\stackrel{*}{T}}(2,2589)(5)=126$ & $\beta_{\stackrel{*}{T}}(2,2589)(9)=65$   & $\beta_{\stackrel{*}{T}}(2,2589)(13)=3$   \\ \midrule
$\beta_{\stackrel{*}{T}}(2,2589)(2)=108$ & $\beta_{\stackrel{*}{T}}(2,2589)(6)=47$  & $\beta_{\stackrel{*}{T}}(2,2589)(10)=241$ & $\beta_{\stackrel{*}{T}}(2,2589)(14)=180$ \\ \midrule
$\beta_{\stackrel{*}{T}}(2,2589)(3)=29$  & $\beta_{\stackrel{*}{T}}(2,2589)(7)=224$ & $\beta_{\stackrel{*}{T}}(2,2589)(11)=162$ & $\beta_{\stackrel{*}{T}}(2,2589)(15)=115$ \\ \midrule
$\beta_{\stackrel{*}{T}}(2,2589)(4)=206$ & $\beta_{\stackrel{*}{T}}(2,2589)(8)=144$ & $\beta_{\stackrel{*}{T}}(2,2589)(12)=83$  & $\beta_{\stackrel{*}{T}}(2,2589)(16)=35$
\\
 \bottomrule
\end{tabular}
\end{table}

{Hence by the respective application of Equation~(\ref{defus}) and the S-box $\zeta_{E_{293}}(293,2,2589)$, the pixel values of diffused image $M_{\rm P}$ and encrypted image $C_{\rm P}$ are shown in Tables~\ref{Mk} and  \ref{En}}, respectively.
\begin{table}[H]
\caption{{Diffused image $M_{\rm P}.$}}
\label{Mk}%
\centering
\begin{tabular}{cccc}
\toprule
94  & 32  & 227 & 166 \\
14  & 209 & 147 & 87  \\
191 & 130 & 68  & 22  \\
110 & 51  & 243 & 194
\\ \bottomrule
\end{tabular}
\end{table}\vspace{-18pt}
\begin{table}[H]
\caption{{Encrypted image $C_{\rm P}.$}}
\label{En}%
\centering
\begin{tabular}{cccc}
\toprule
76  & 231 & 254 & 19  \\
194 & 54  & 161 & 65  \\
0   & 67  & 162 & 209 \\
151 & 69  & 34  & 1
\\ \bottomrule
\end{tabular}
\end{table}

\subsection{Decryption} In our scheme the decryption process can take place by reversing the operations of the encryption process. One should know the inverse S-box $\zeta^{-1}_{E_{p}}(n,t,S_{\rm P})$ and the pseudo-random numbers $\beta_{\stackrel{*}{T}}(t,S_{\rm P})$. Assume the situation when the secret keys $a_{1}, b_{1}, a_{2}, b_{2}, a_{3},\delta$, $L,S_{\rm P},t$ and $p$ are transmitted by a secure channel, so that the set $\stackrel{*}{T}$ is obtained using keys $a_{1}, b_{1}, a_{2}, b_{2}, a_{3},\delta$ and $L$, and hence the S-box $\zeta^{-1}_{E_{p}}(p,t,S_{\rm P})$ and the pseudo-random numbers $\beta_{\stackrel{*}{T}}(t,S_{\rm P})$ can be computed by  $S_{\rm P},t$ and $p$. Finally, the~receiver gets the original image $P$ by applying the following equations:
\begin{align}
&M_{\rm P}(i)=\zeta_{E_{p}}^{-1}(p,t,S_{\rm P})(C_{\rm P}(i)),\\
\begin{split}
P(i)=M_{\rm P}(i)-\beta_{\stackrel{*}{T}}(t,S_{\rm P})(i)\pmod{256}.
\end{split}
\end{align}
\section{Security Analysis}\label{SEAnalysis}
{In this section the cryptographic strength of both the S-box construction technique and encryption scheme are analyzed in detail.}
\subsection{Evaluation of the Designed S-Box}\label{Sb1}
{An S-box with good cryptographic properties ensures the quality of an encryption technique. Generally, some standard tests such as nonlinearity (NL), linear approximation probability (LAP), strict avalanche criterion (SAC), bit independence criterion (BIC) and differential approximation probability (DAP) are used to evaluate the cryptographic strength of an S-box.}

{The NL~\cite{Ad90} and the LAP~\cite{M94} are outstanding features of an S-box, used to measure the resistance against linear attacks. The NL measures the level of nonlinearity and the LAP finds the maximum imbalance value of an S-box. The optimal value of the nonlinearity is $112$. A low value of LAP corresponds to a high resistance. The minimum NL and the LAP values for the displayed S-box are $106$ and $0.1484$, respectively. This ensures that the proposed S-box is immune to  linear attacks. Webster and Tavares~\cite{We85} developed the concepts of the SAC and the BIC, which are used to find the confusion and diffusion creation potential of an S-box. In other words, the SAC criterion measures the change in output bits when an input bit is altered. Similarly, the BIC criterion explores the correlation in output bits when change in a single input bit occurs. The average values of the SAC and the BIC for the constructed S-box are $0.4951$  and $0.4988$, respectively, which are close to the optimal value $0.5$. Thus, both tests are satisfied by the suggested S-box. The DAP~\cite{Bi91} is another important feature used to analyze the capability of an S-box against differential attacks. The lowest value of DAP for an S-box implies the highest security to the differential attacks. Our DAP result is $0.0234$, which is good enough to resist differential cryptanalysts.}
\subsection{Performance Comparison of the S-Box Generation Algorithm}\label{Sb2}
\textls[-15]{{After performing the rigorous analyses, the S-box constructed by the current algorithm is compared with some cryptographically strong S-boxes developed by recent schemes, as shown in Table~\ref{Scom}.}}
\begin{table}[H]
\caption{{Comparison table of the proposed S-box $\zeta_{E_{1607}}(1607,182,0)$}.}
\label{Scom}
\centering
\resizebox{\columnwidth}{!}
{
\bgroup
\def\arraystretch{1.1}
\begin{tabular}{cccccccccc}
    \toprule
\textbf{S-Boxes}&\textbf{NL}&\textbf{LAP}& \multicolumn{3}{c}{\textbf{SAC}}& \multicolumn{3}{c}{\textbf{BIC}}&\textbf{DAP}\\
\cmidrule{4-9}
       &  &   & \textbf{(min)}& \textbf{(avg)} &\textbf{(max)}&\textbf{(min)}& \textbf{(avg)} &\textbf{(max)}  & \\ \midrule
Ours& 106 & 0.1484375 & 0.390625 & 0.49511719 & 0.609375 & 0.47265625 & 0.49888393 & 0.52539063 & 0.0234375\\
Ref.~\cite{signal}& 104   & 0.1484375 &0.421900          &   -       & 0.6094    &0.4629     &            -    &  0.5430         & 0.0469 \\
Ref.~\cite{CM1} & 104 & 0.1328125 & 0.40625  & 0.49755859 & 0.625    & 0.46679688 & 0.50223214 & 0.5234375  & 0.0234375 \\
Ref.~\cite{CM2}  & 101 & 0.140625  & 0.421875 & 0.49633789 & 0.578125 & 0.46679688 & 0.49379185 & 0.51953125 &  0.03125\\
Ref.~\cite{CM3}  & 104 & 0.140625  & 0.421875 & 0.50390625 & 0.59375  & 0.4765625  & 0.50585938 & 0.5390625  &0.0234375  \\
Ref.~\cite{CM4}  & 100 & 0.140625  & 0.40625  & 0.50097656 & 0.609375 & 0.44726563 & 0.50634766 & 0.53320313 & 0.03125 \\
Ref.~\cite{CM5}  & 106 & 0.140625  & 0.390625 & 0.49414063 & 0.609375 & 0.47070313 & 0.50132533 & 0.53320313 & 0.0234375 \\
Ref.~\cite{CM6}  & 102 & 0.140625  & 0.421875 & 0.49804688 & 0.640625 & 0.4765625  & 0.50746373 & 0.53320313 & 0.0234375 \\
Ref.~\cite{CM7} &104     &0.0391        &0.3906       &  -          &0.6250         &0.4707            & -           & 0.53125           & 0.0391 \\
Ref.~\cite{CM8} & 104  & 0.0547000          & 0.4018     & 0.4946        & 0.5781     &0.4667969            & 0.4988839           &  0.5332031          & 0.0391 \\
Ref.~\cite{Alzaidi}&108&0.1328  &0.40625&0.4985352&0.59375 &0.46484375&0.5020229&0.52734375&0.0234375\\ \bottomrule
\end{tabular}
\egroup
}
\end{table}
{From Table~\ref{Scom} it follows that the NL of $\zeta_{E_{1607}}(1607,182,0)$ is greater than the S-boxes in~\cite{signal,CM1,CM2,CM3,CM4,CM6,CM7,CM8}, equal to that of~\cite{CM5} and less than the S-box developed in~\cite{Alzaidi}, which indicates that $\zeta_{E_{1607}}(1607,182,0)$ is highly nonlinear in comparison to the S-boxes in~\cite{signal,CM1,CM2,CM3,CM4,CM6,CM7,CM8}. Additionally, the LAP of $\zeta_{E_{1607}}(1607,182,0)$ is comparable to all the S-boxes in Table~\ref{Scom}. The SAC (average) value of $\zeta_{E_{1607}}(1607,182,0)$ is greater than the S-boxes in~\cite{CM5,CM8}, and the SAC (max) value is less than or equal to the S-boxes in~\cite{signal,CM1,CM4,CM5,CM6,CM7}. Similarly the BIC (min) value of $\zeta_{E_{1607}}(1607,182,0)$ is closer to the optimal value $0.5$ than that of~\cite{signal,CM1,CM2,CM4,CM5,CM7,CM8,Alzaidi}, and the BIC (max) value of the new S-box is better than that of the S-boxes in~\cite{signal,CM3,CM4,CM5,CM6,CM7,CM8,Alzaidi}. Thus the confusion/diffusion creation capability of $\zeta_{E_{1607}}(1607,182,0)$ is better than ~\cite{signal,CM4,CM5,CM6,CM7,Alzaidi}. The DAP value of our suggested S-box $\zeta_{E_{1607}}(1607,182,0)$ is lower than the DAP of the S-boxes presented in~\cite{signal, CM2,CM4,CM7,CM8} and equal to that of~\cite{CM1,CM3,CM5,CM6,Alzaidi}. Thus from the above discussion it follows that the newly designed S-box shows high resistance to linear as well as differential attacks.}
\subsection{Evaluation of the Proposed Encryption Technique}\label{Analysis}
\textls[-15]{In this section
 the current scheme is implemented on all gray images of the USC-SIPI Image Database~\cite{Dbase}. The USC-SIPI database contains images of size $m \times m$, $m$ = 256,512,1024. Furthermore, some security analyses that are explained one by one in the associated subsections are presented. To validate the quality of the proposed scheme, the experimental results are compared with some other encryption schemes. The parameters used for the experiments are $A_{1}=A_{2}=-1.0541, A_{3}=401, B_{1}=B_{2}=-0.8514$ and $B_{3}=691,3036,5071$ \rm{for} $m$ = 256,512,1024, \rm{respectively}; $a_{1}=2,b_{1}=1000, a_{2}=19,b_{2}=1000, a_{3}=5,\delta=1000,t=2,p=293,L$ = 90,000 and $S_{\rm P}$ varies for each $P$. The experiments were performed using Matlab R$2016$a on a personal computer with a $1.8$ GHz Processor and $6$ GB RAM. All encrypted images of the database along with histograms are available at~\cite{github}. Some plain images, House$_{256\times 256}$, Stream$_{512\times 512}$, Boat$_{512\times 512}$ and Male$_{1024\times 1024}$ and their cipher images are displayed in Figure~\ref{fig:encryp}.}
\begin{figure}[H]
\captionsetup[subfigure]{justification=centering}
\centering
\begin{subfigure}[b]{0.24\textwidth}
	\centering
	 \includegraphics[scale=0.198]{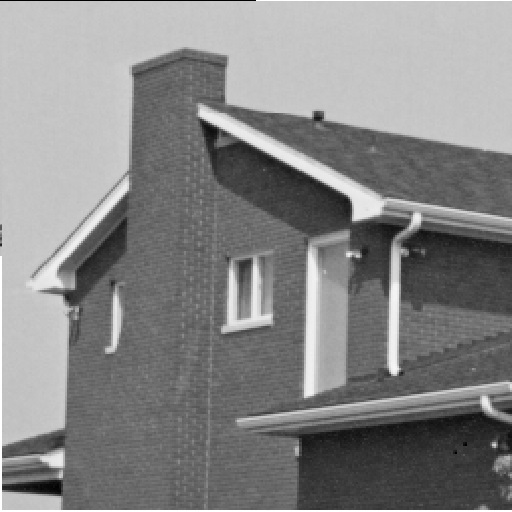}
	 \caption{ }
	 \end{subfigure}
\begin{subfigure}[b]{0.24\textwidth}
	\centering
	 \includegraphics[scale=0.198]{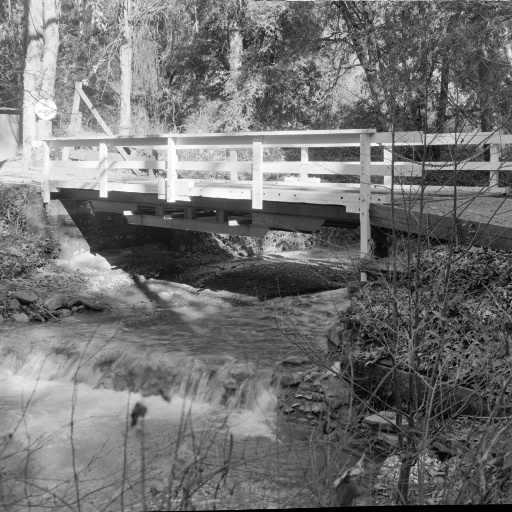}	
      \caption{ }
	 \end{subfigure}
	\begin{subfigure}[b]{0.24\textwidth}
	\centering
	 \includegraphics[scale=0.198]{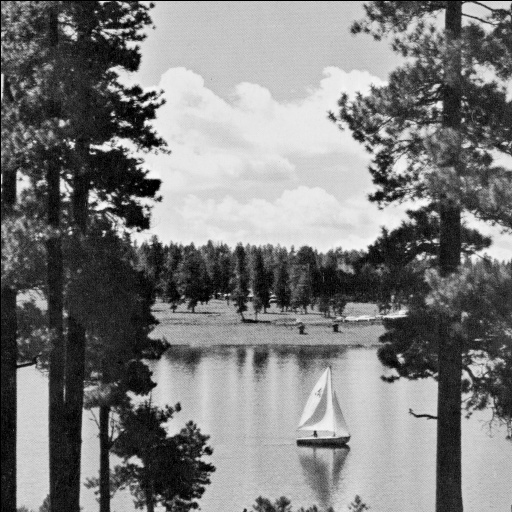}
	 \caption{}
	 \end{subfigure}
	 \begin{subfigure}[b]{0.24\textwidth}
	\centering
	 \includegraphics[scale=0.198]{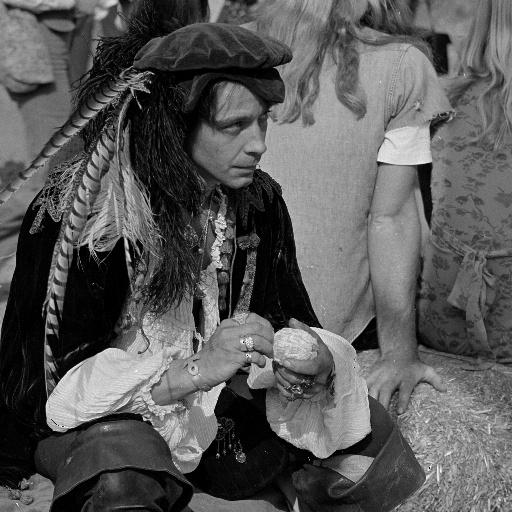}
	 \caption{ }
	 \end{subfigure}
\bigskip\\
\begin{subfigure}[b]{0.24\textwidth}
	\centering
	 \includegraphics[scale=0.198]{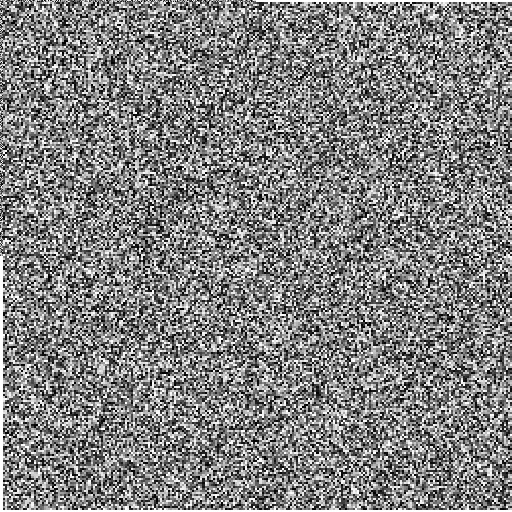}
	 \caption{}
	 \end{subfigure}
\begin{subfigure}[b]{0.24\textwidth}
	\centering
	 \includegraphics[scale=0.198]{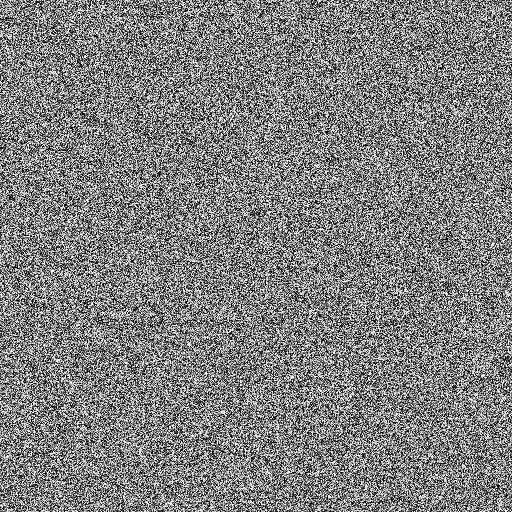}
	 \caption{ }
	 \end{subfigure}
	 \begin{subfigure}[b]{0.24\textwidth}
	\centering
	 \includegraphics[scale=0.198]{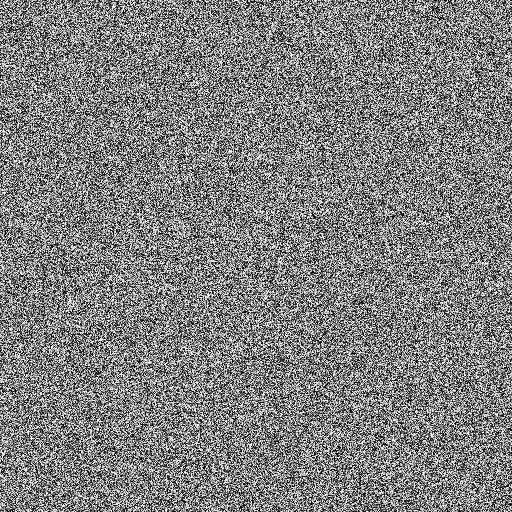}
	 \caption{}
	 \end{subfigure}
\begin{subfigure}[b]{0.24\textwidth}
	\centering
	 \includegraphics[scale=0.198]{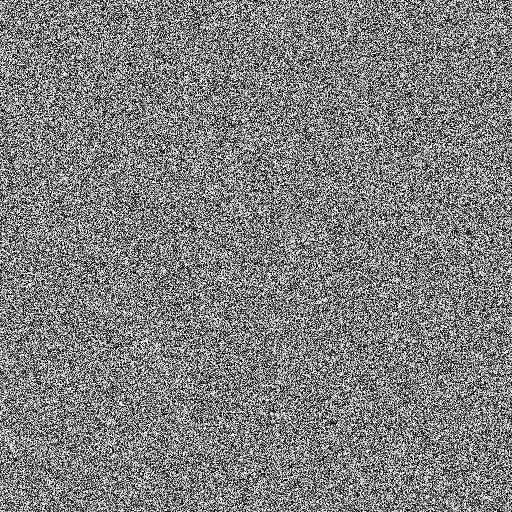}
	 \caption{ }
	 \end{subfigure}

\caption{(\textbf{a})–(\textbf{d})~Plain images House, Stream, Boat and Male;
(\textbf{e})–(\textbf{h})~cipher images of the plain images~(\textbf{a})–(\textbf{d}), respectively.
 }
\label{fig:encryp}
\end{figure}
\subsubsection{Statistical Attack}
A cryptosystem is said to be secure if it has high resistance against statistical attacks. The strength of resistance against statistical attacks is measured by entropy, correlation and histogram tests. All of these tests are applied to evaluate the performance of the discussed scheme.
\begin{enumerate}
\item[(1)] \textls[-15]{Histogram. A histogram is a graphical way to display the frequency distribution of pixel values of an image. A secure cryptosystem generates cipher images with uniform histograms. The histograms of the encrypted images using the proposed method are available at~\cite{github}.
 However, the respective histograms for the images in Figure~\ref{fig:encryp} are shown in Figure~\ref{fig:Hist}. The histograms of the encrypted images are almost uniform. Moreover, the histogram of an encrypted image is totally different from that of the respective plain image, so that it does not allow useful information to the adversaries, and the proposed algorithm can resist any statistical attack.}
    \begin{figure}[H]
\captionsetup[subfigure]{justification=centering}
\centering
\begin{subfigure}[b]{0.24\textwidth}
	\centering
	 \includegraphics[scale=0.28]{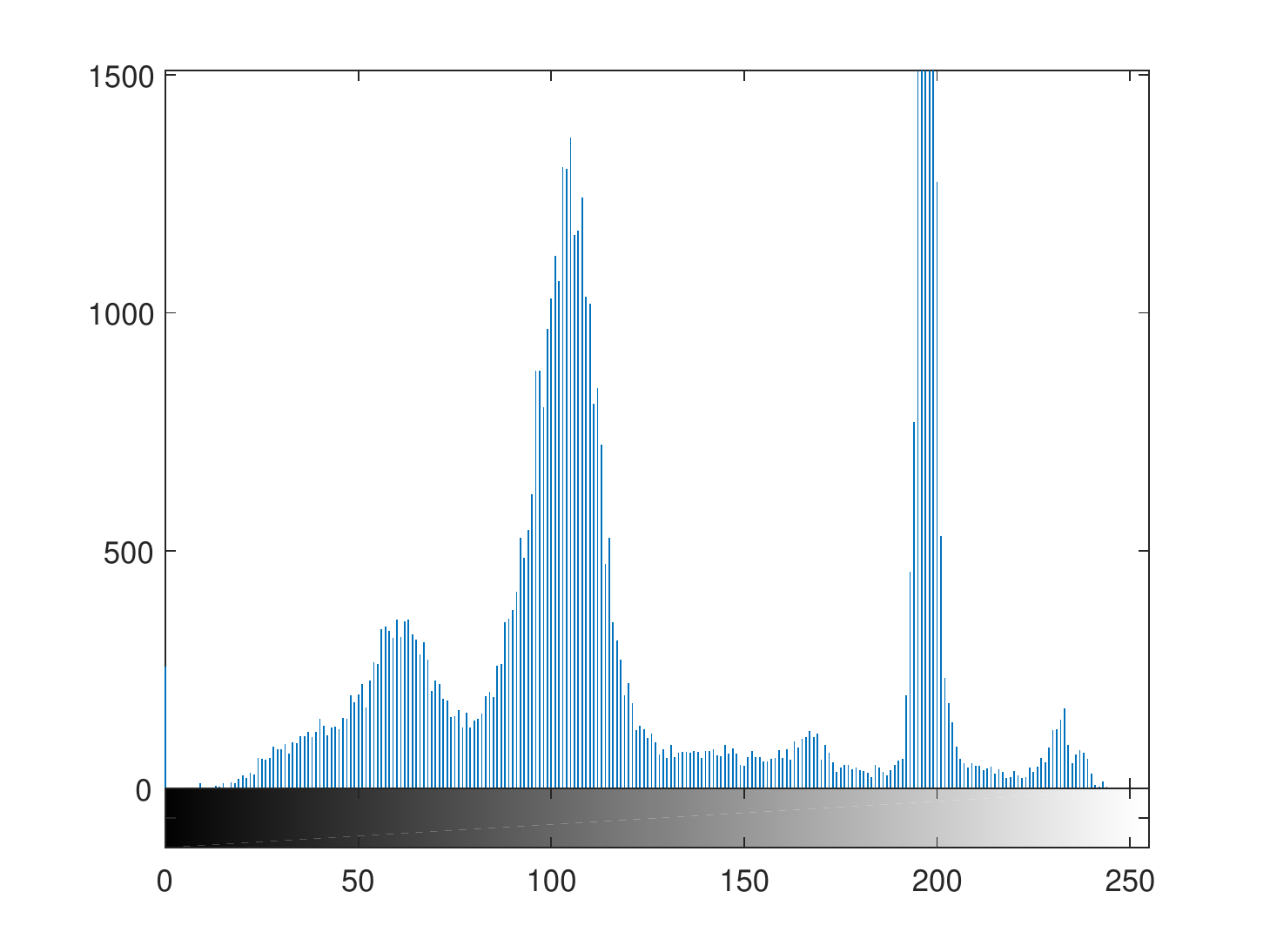}
	 \caption{ }
	 \end{subfigure}
\begin{subfigure}[b]{0.24\textwidth}
	\centering
	 \includegraphics[scale=0.28]{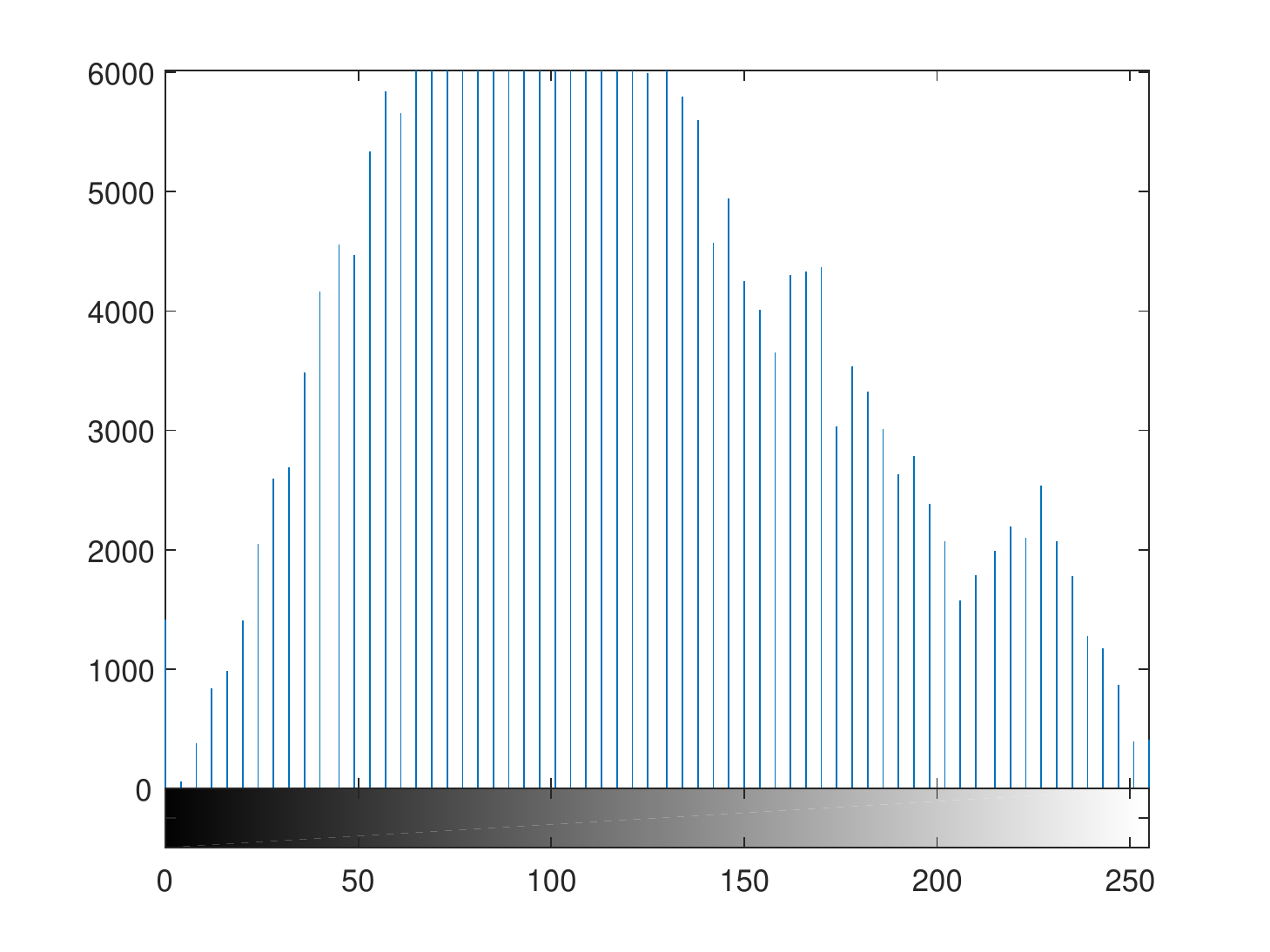}	
      \caption{ }
	 \end{subfigure}
	\begin{subfigure}[b]{0.24\textwidth}
	\centering
	 \includegraphics[scale=0.28]{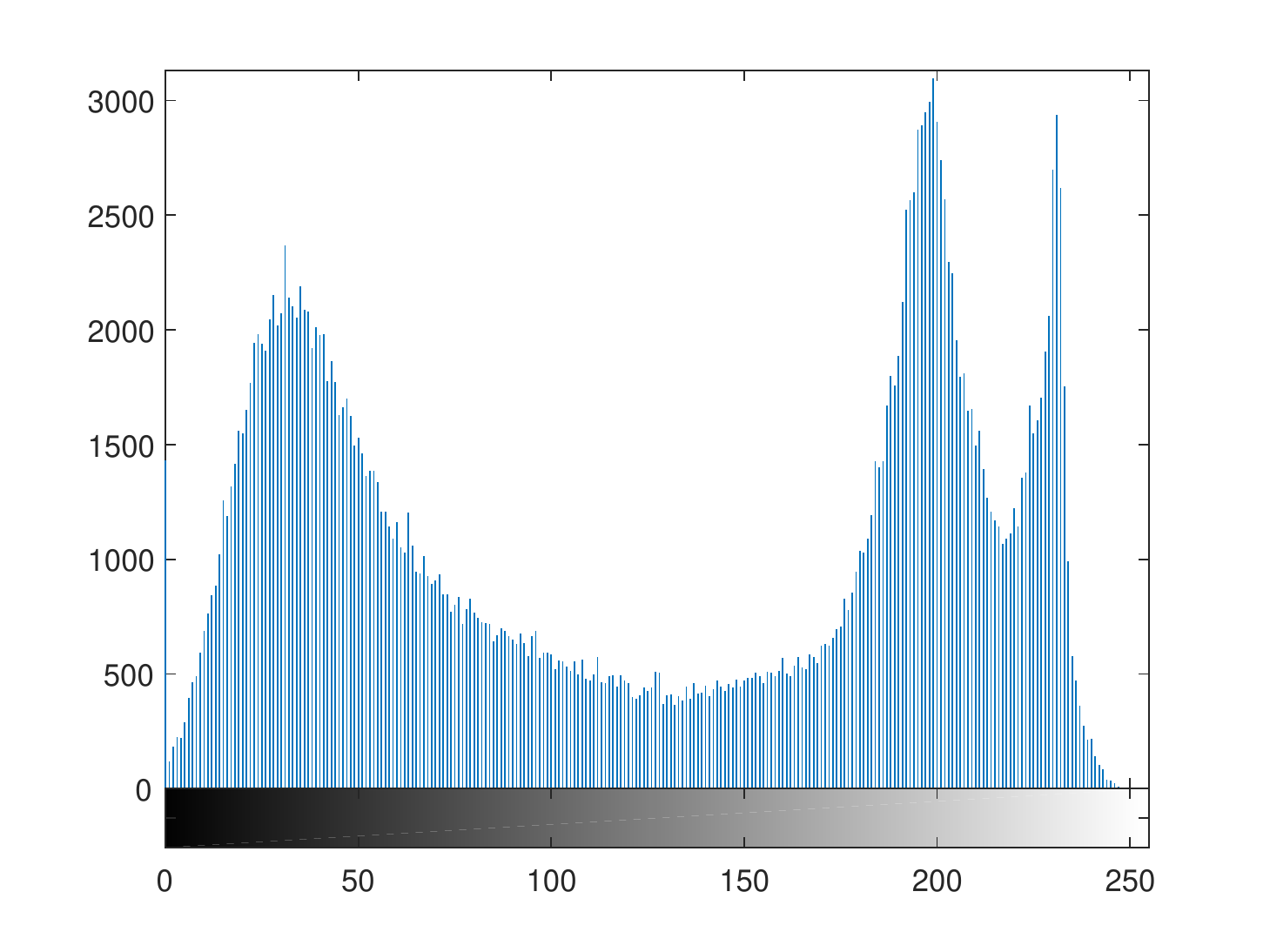}
	 \caption{}
	 \end{subfigure}
	 \begin{subfigure}[b]{0.24\textwidth}
	\centering
	 \includegraphics[scale=0.28]{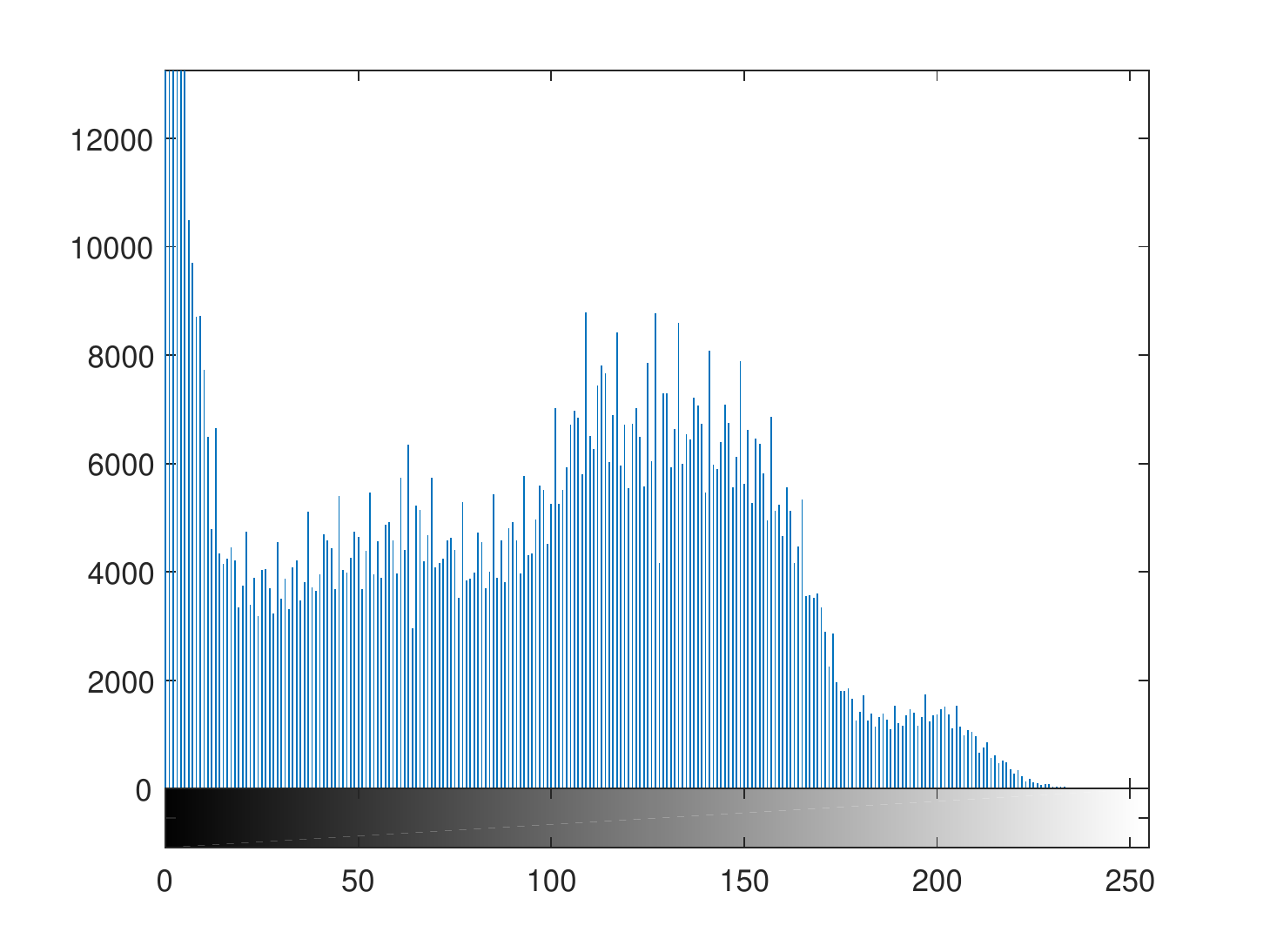}
	 \caption{ }
	 \end{subfigure}
\bigskip\\
\begin{subfigure}[b]{0.24\textwidth}
	\centering
	 \includegraphics[scale=0.28]{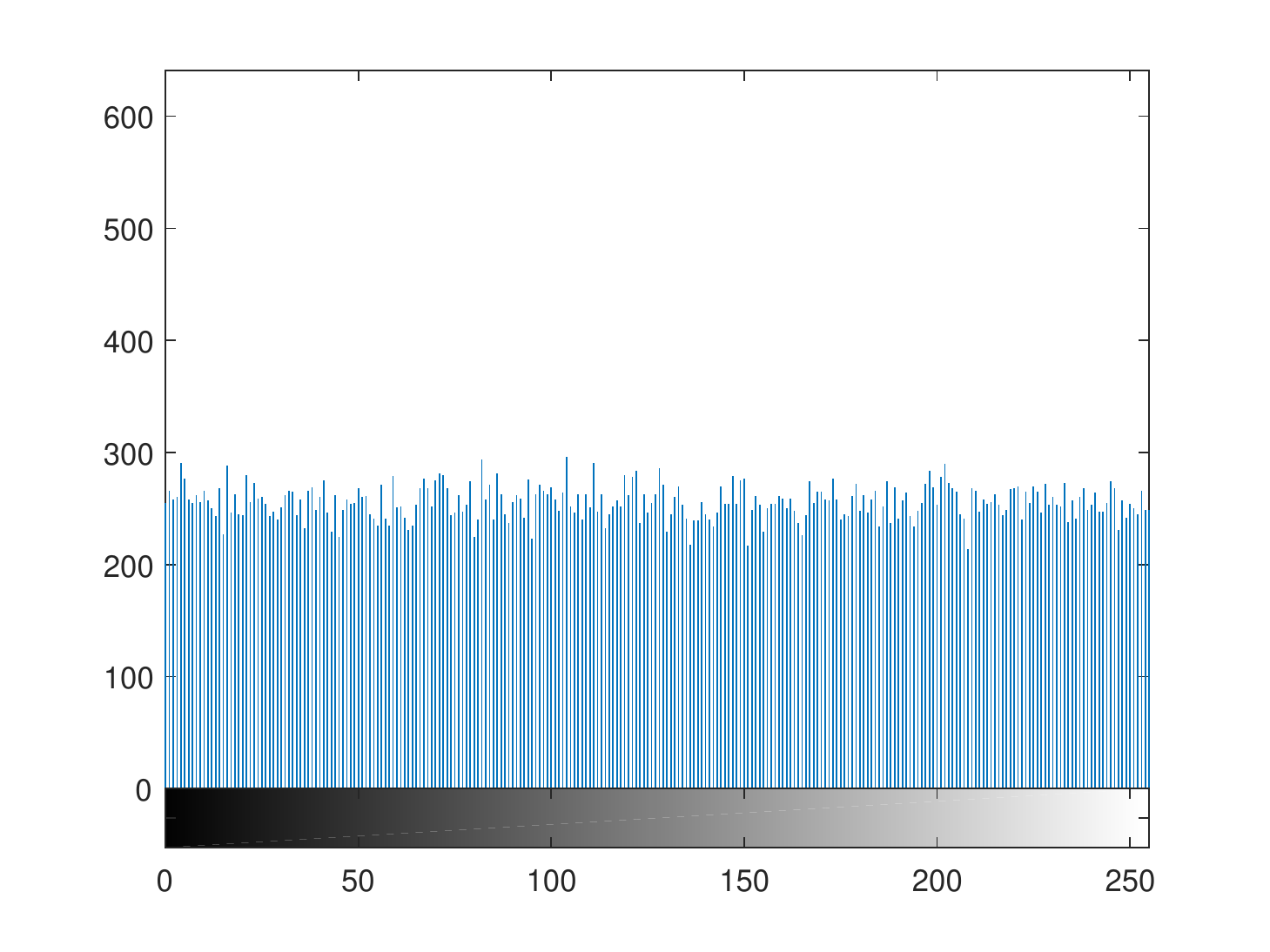}
	 \caption{}
	 \end{subfigure}
\begin{subfigure}[b]{0.24\textwidth}
	\centering
	 \includegraphics[scale=0.28]{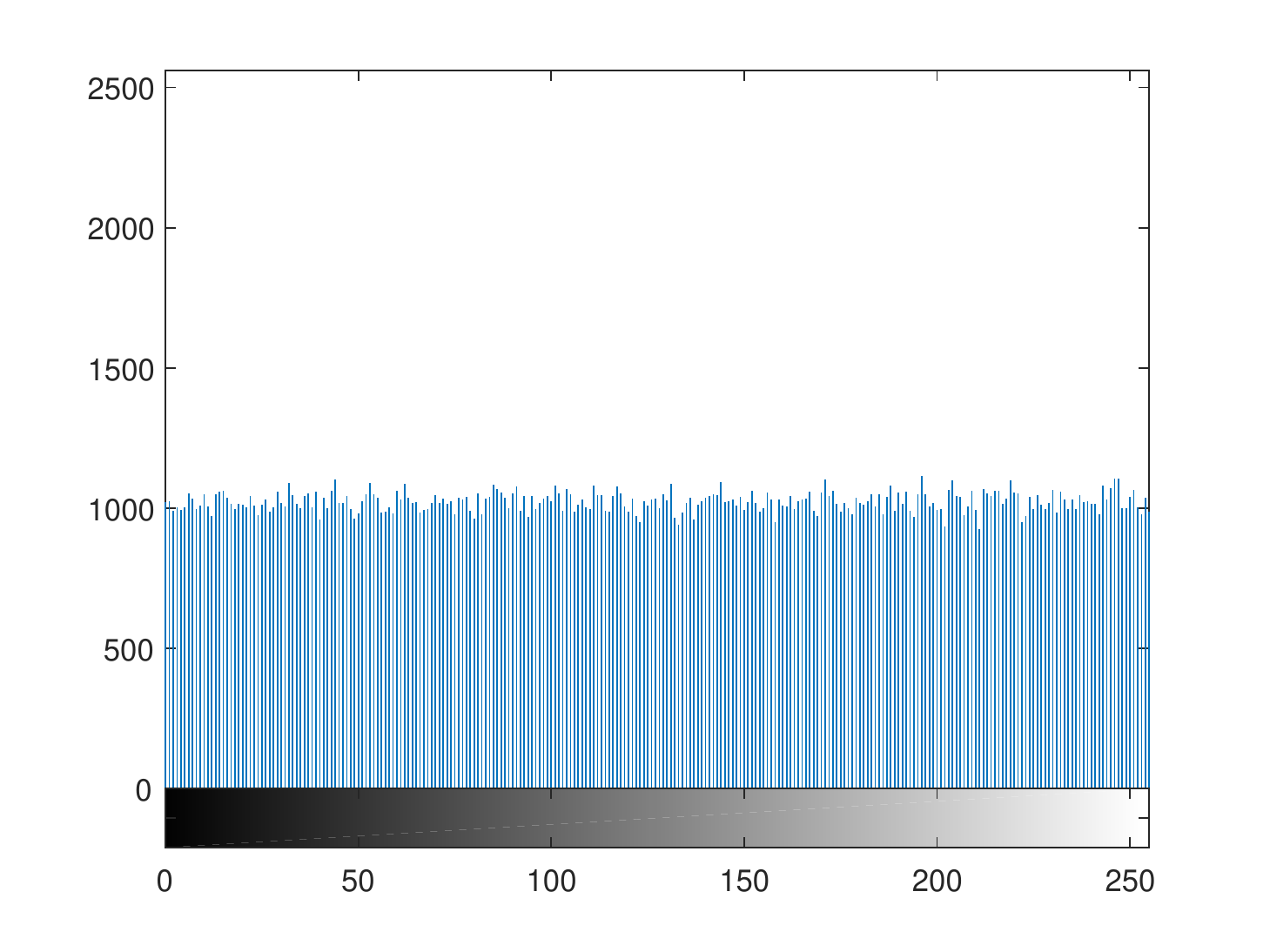}
	 \caption{ }
	 \end{subfigure}
	 \begin{subfigure}[b]{0.24\textwidth}
	\centering
	 \includegraphics[scale=0.28]{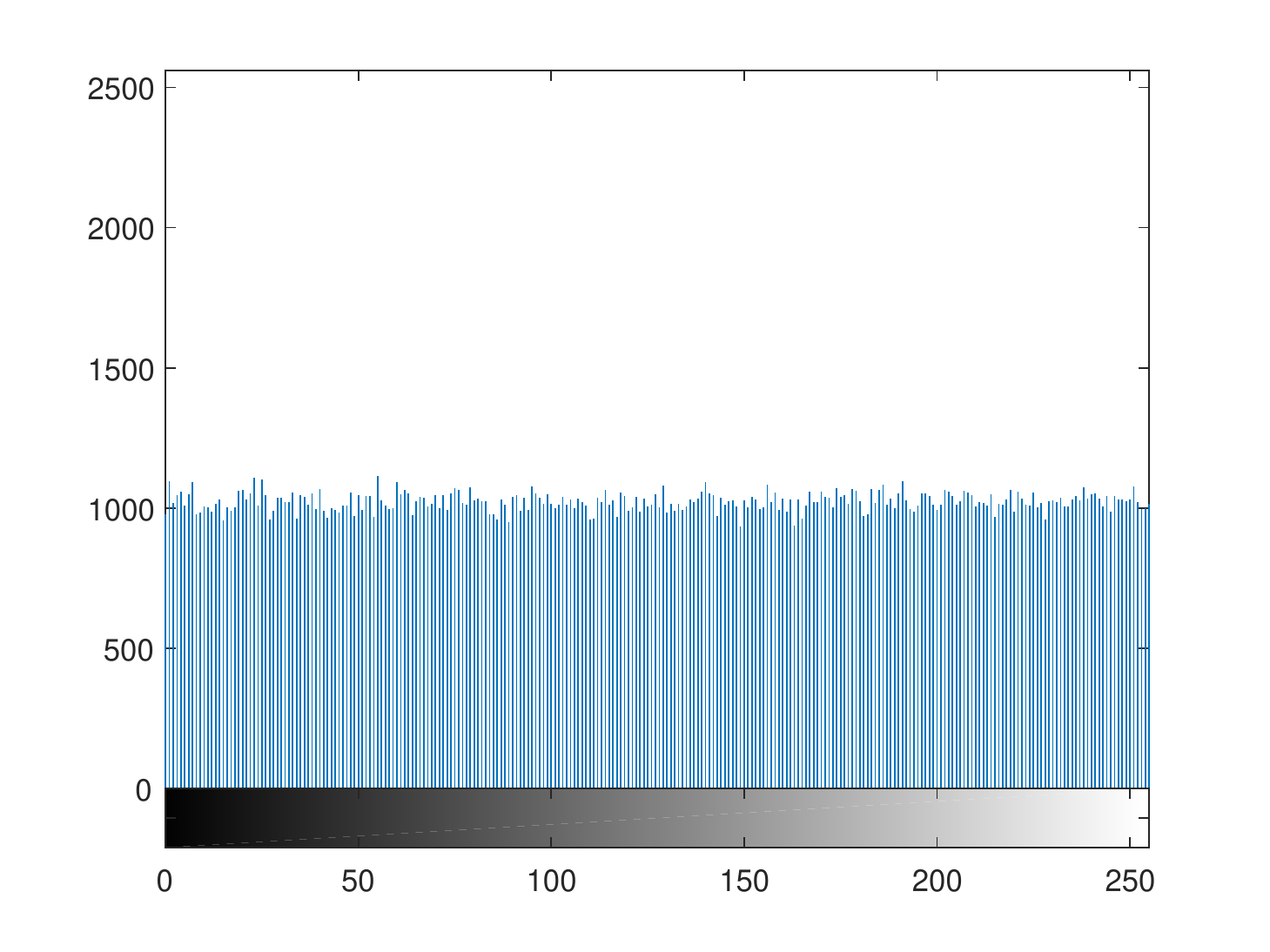}
	 \caption{}
	 \end{subfigure}
\begin{subfigure}[b]{0.24\textwidth}
	\centering
	 \includegraphics[scale=0.28]{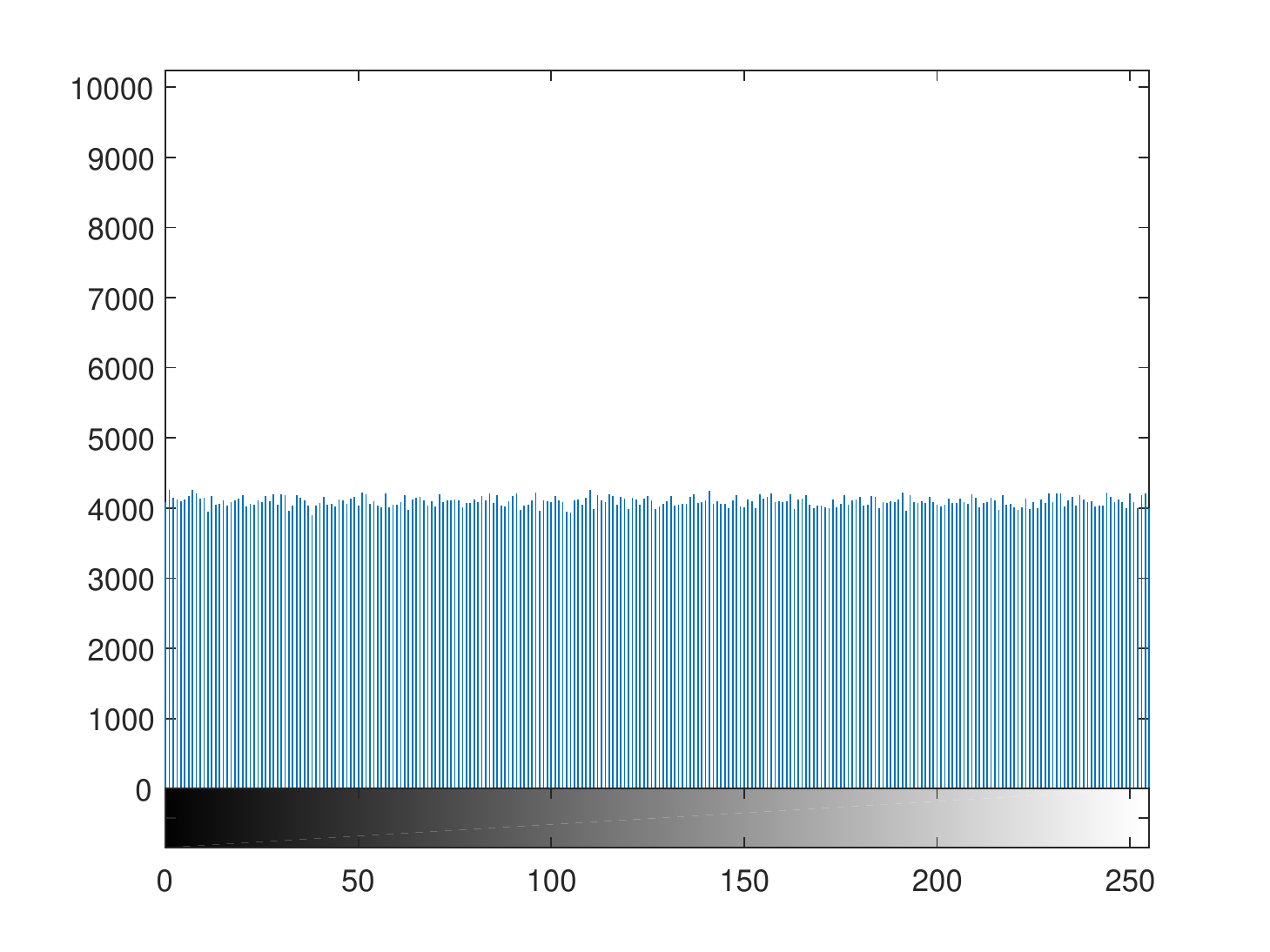}
	 \caption{ }
	 \end{subfigure}

\caption{(\textbf{a})--(\textbf{d})~Histograms of Figure~\ref{fig:encryp}(\textbf{a})--(\textbf{d});
(\textbf{e})--(\textbf{h}) histograms of Figure~\ref{fig:encryp}(\textbf{e})--(\textbf{h}).}
\label{fig:Hist}
\end{figure}
\item[(2)] Entropy. Entropy is a standout feature to measure the disorder. Let $I$ be a source of information over a set of symbols $N$. Then the entropy of $I$ is defined by:
    \begin{equation}
    {\rm H}(I) = \sum_{i=1}^{\#N} p(I_{i}){\rm log}_{2}\frac{1}{p(I_{i})},
    \end{equation}
where $p(I_{i})$ is the probability of occurrence of symbol $i.$ The ideal value of ${\rm H}(I)$ is ${\rm log}_{2}{(\#N)}$, if~all symbols of $N$ occur in $I$ with the same probability. Thus, an image $I$ emanating $256$ gray levels is highly random if ${\rm H}(I)$ is close to $8$ (notice, however, that this definition of entropy does not take into account pixel correlations). The entropy results for all images encrypted by the suggested technique are shown in Figure~\ref{fig:Corr}, where the minimum, average and maximum values are $7.9966,7.9986$ and $7.9999$, respectively. These results are close to $8$, and hence the developed mechanism is secure against entropy attacks.
\item[(3)] Pixel correlation. A meaningful image has strong correlation among the adjacent pixels. In fact, a~good cryptosystem has the ability to break the pixel correlation and bring it close to zero. For any two gray values $x$ and $y$, the pixel correlation can be computed as:
    \begin{equation}
    C_{xy}=\frac{E\big[(x-E[x])(y-E[y])\big]}{\sqrt{K[x]K[y]}},
    \end{equation}
where $E[x]$ and $K[x]$ denote expectation and variance of $x$, respectively. The range of $C_{xy}$ is $-1$ to $1$. The gray values $x$ and $y$ are in low correlation if $C_{xy}$ is close to zero. As the pixels may be adjacent in horizontal, diagonal and vertical directions, the correlation coefficients of all encrypted images along all  three directions are shown in Figure~\ref{fig:Corr}, where the respective ranges of $C_{xy}$ are [$-0.0078$, $0.0131$], [$-0.0092$,$0.0080$] and [$-0.0100$,$0.0513$]. These results show that the presented method is capable of reducing the pixel correlation near to zero.
\begin{figure}[H]
\captionsetup[subfigure]{justification=centering}
\centering
\begin{subfigure}[b]{0.24\textwidth}
	\centering
	 \includegraphics[scale=0.275]{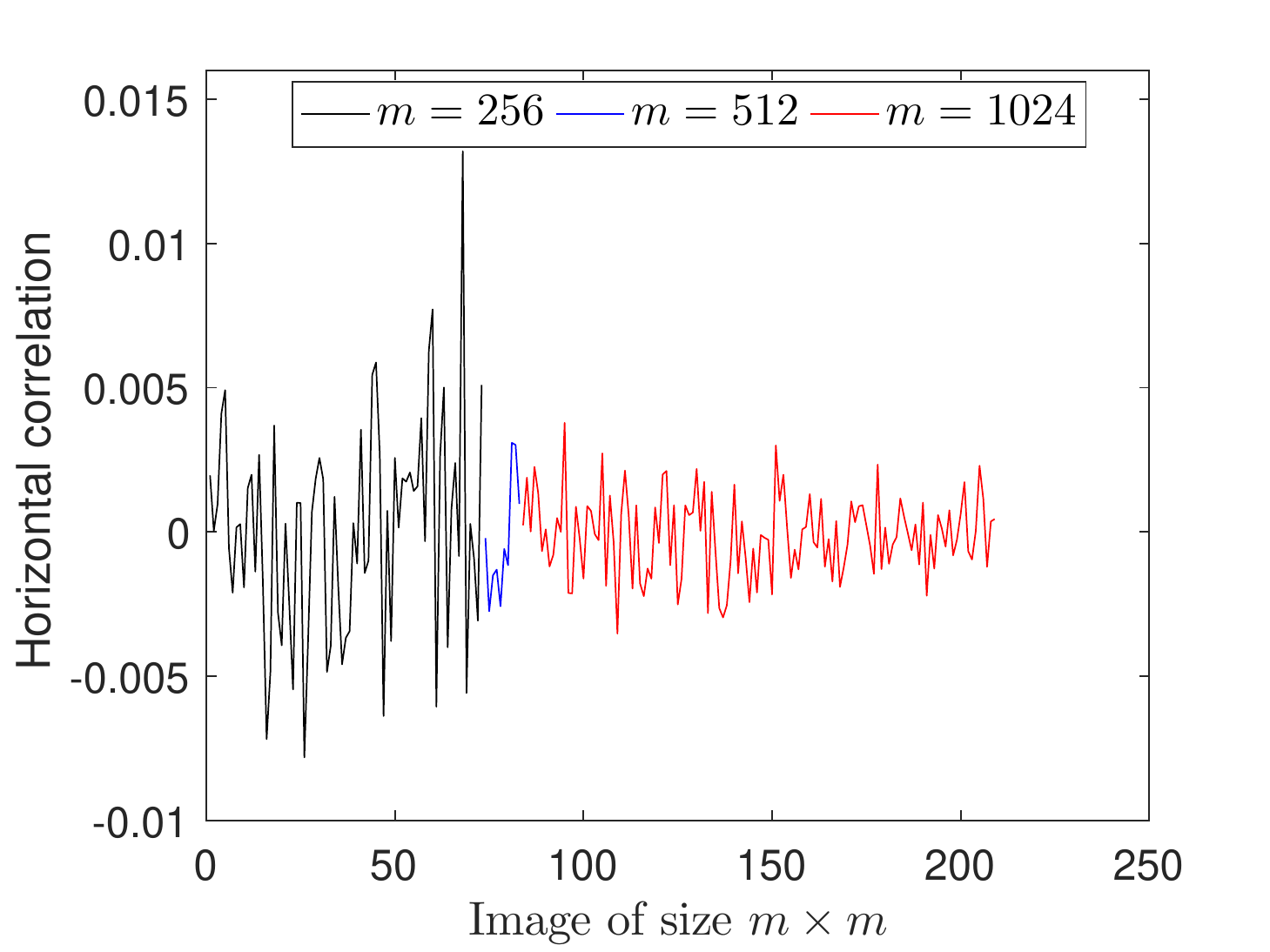}
	 \caption{ }
	 \end{subfigure}
\begin{subfigure}[b]{0.24\textwidth}
	\centering
	 \includegraphics[scale=0.275]{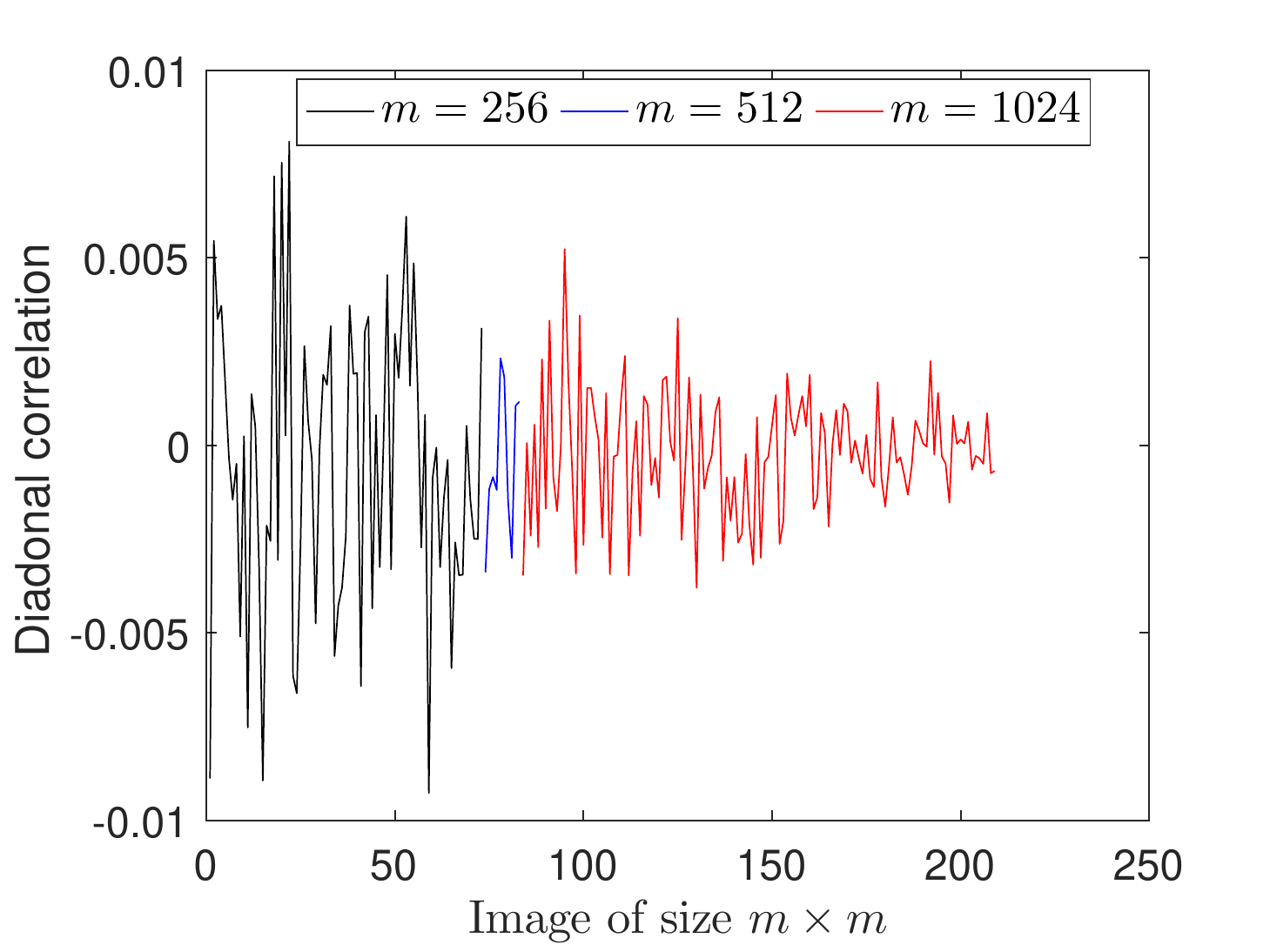}	
      \caption{ }
	 \end{subfigure}
	\begin{subfigure}[b]{0.24\textwidth}
	\centering
	 \includegraphics[scale=0.275]{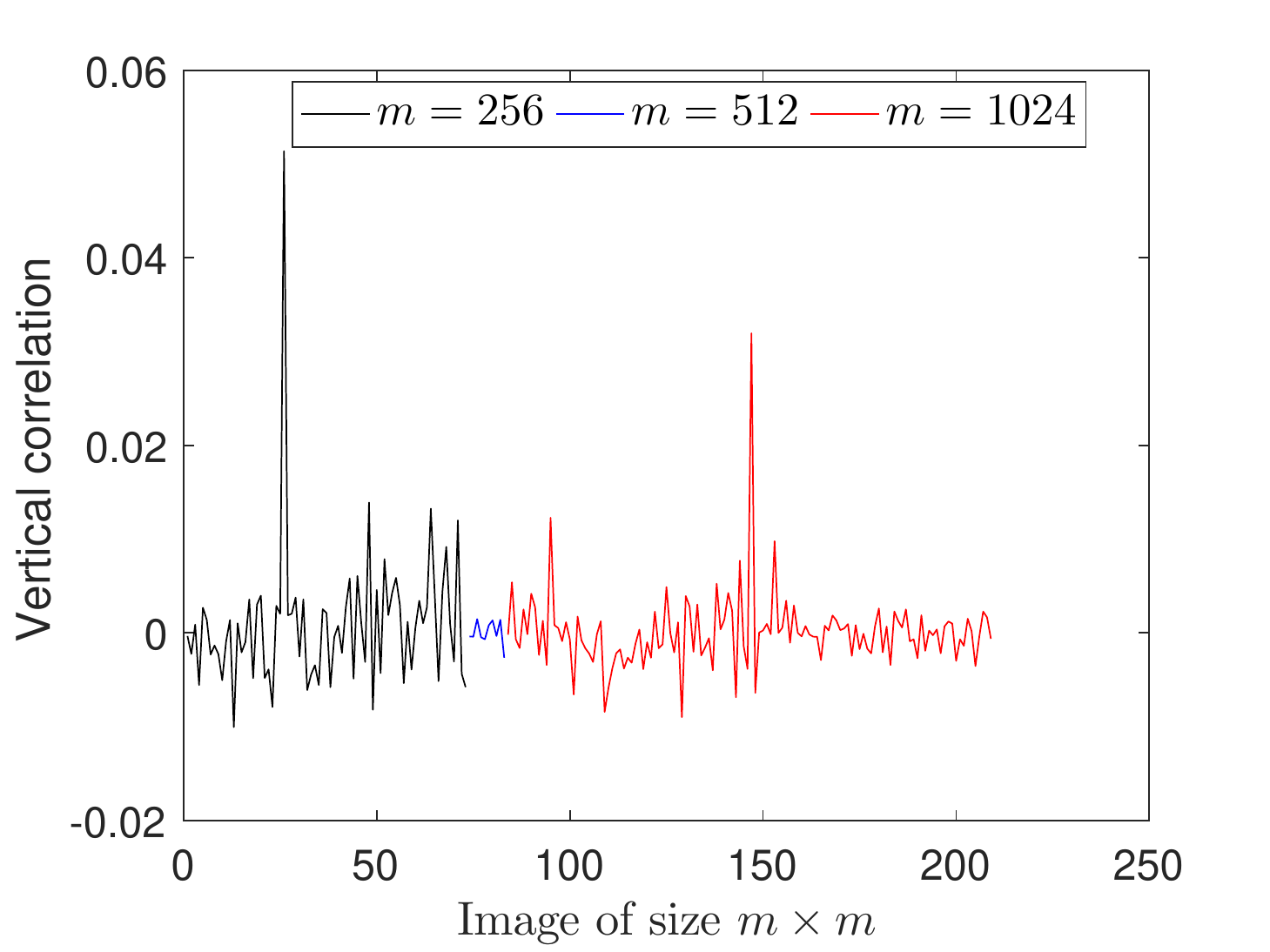}
	 \caption{}
	 \end{subfigure}
	 \begin{subfigure}[b]{0.24\textwidth}
	\centering
	 \includegraphics[scale=0.275]{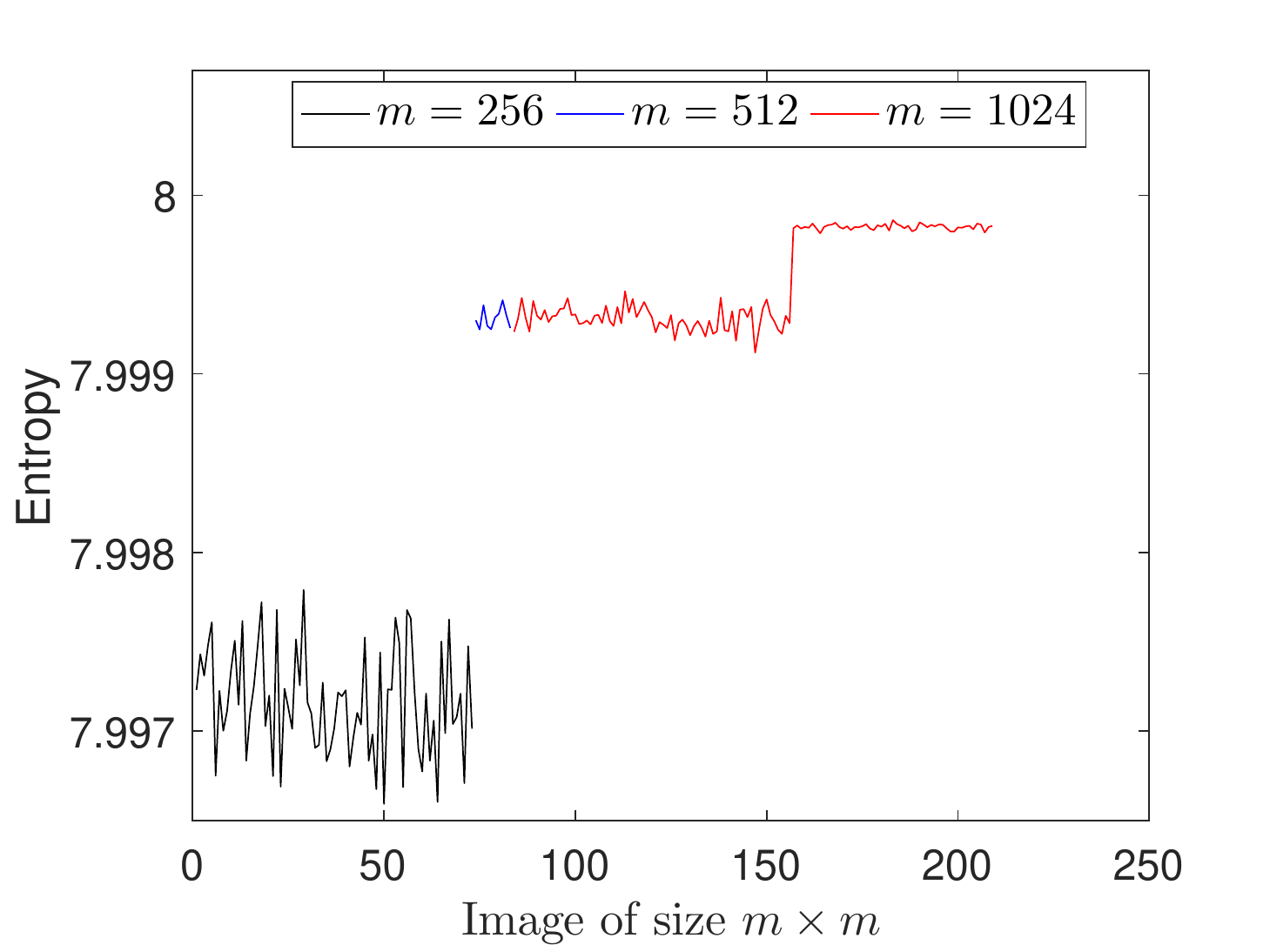}
	 \caption{ }
	 \end{subfigure}
\caption{(\textbf{a})–(\textbf{c})~The horizontal, diagonal and vertical correlations among pixels of each image in USC-SIPI database; (\textbf{d})~the entropy of each image in USC-SIPI database.}
\label{fig:Corr}
\end{figure}

In addition, $2000$ pairs of adjacent pixels of the plain image and cipher image of Lena$_{512\times 512}$ are randomly selected. Then correlation distributions of the adjacent pixels in all  three directions are shown in Figure~\ref{fig:scatter}, which reveals the strong pixel correlation in the plain image but a weak pixel correlation in the cipher image generated by the current scheme.
\begin{figure}[H]
\captionsetup[subfigure]{justification=centering}
\centering
\begin{subfigure}[b]{0.21\textwidth}
	\centering
	 \includegraphics[scale=0.21]{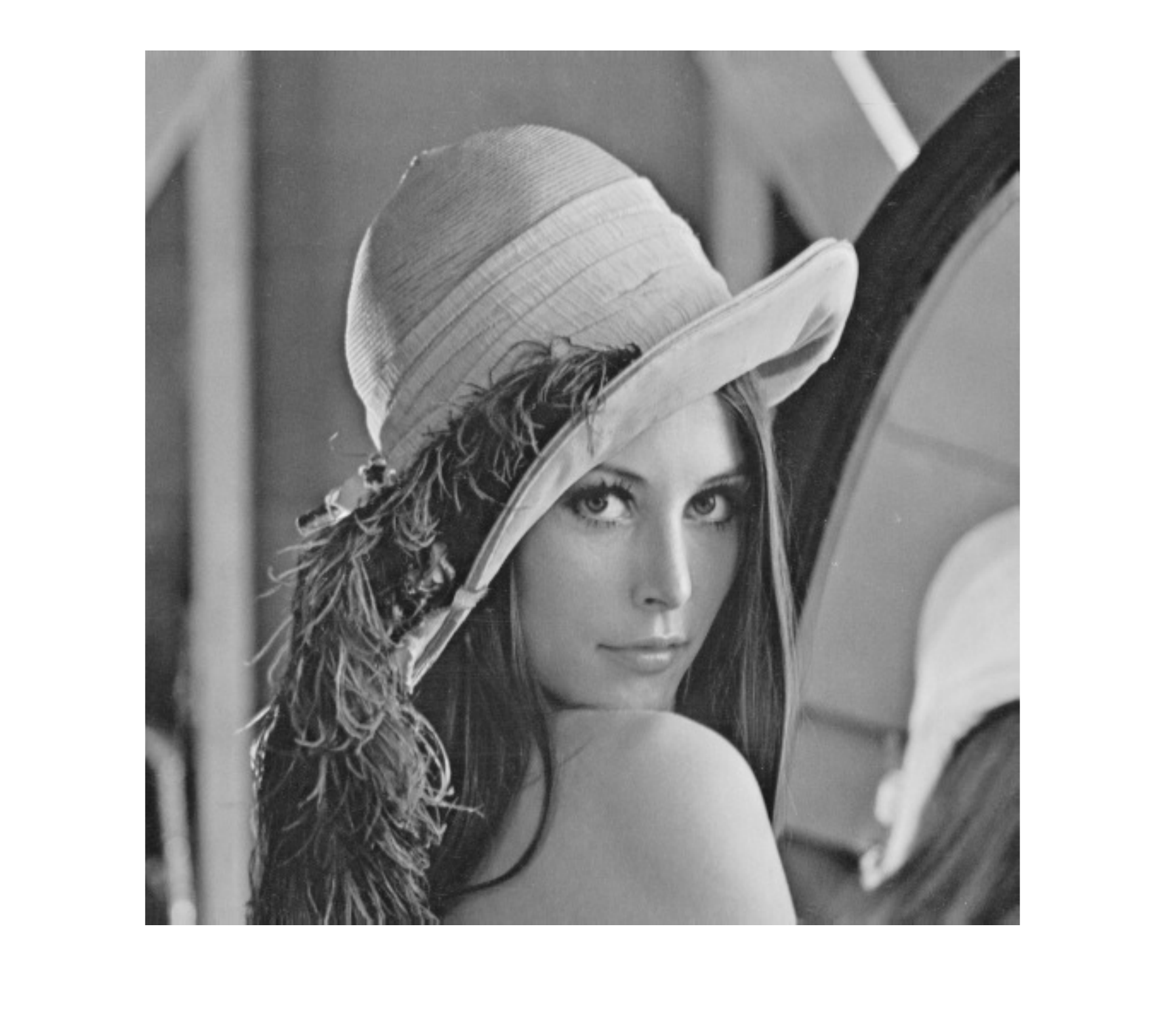}
	 \caption{ }
	 \end{subfigure}
\begin{subfigure}[b]{0.25\textwidth}
	\centering
	 \includegraphics[scale=0.29]{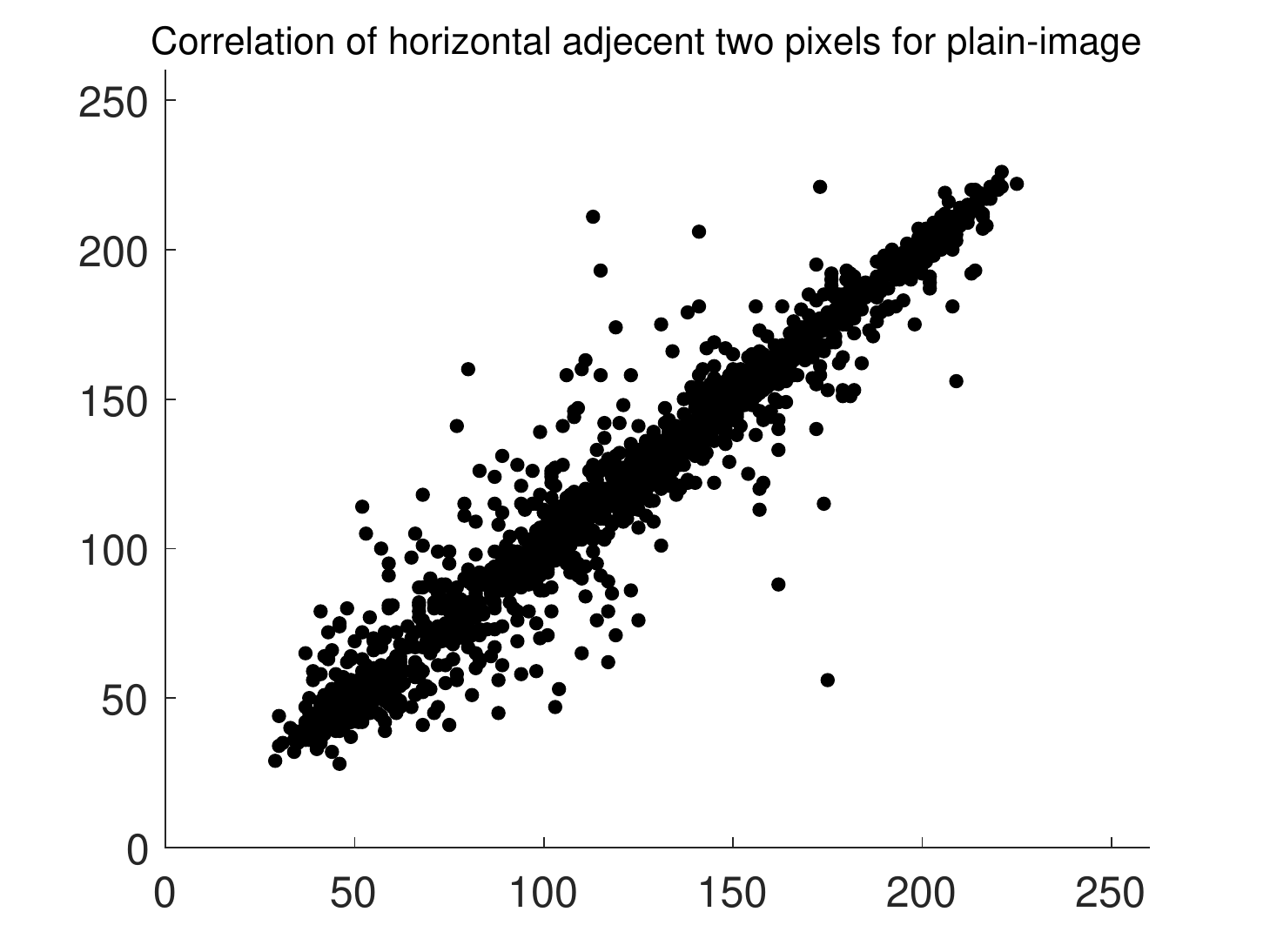}	
      \caption{ }
	 \end{subfigure}
	\begin{subfigure}[b]{0.25\textwidth}
	\centering
	 \includegraphics[scale=0.29]{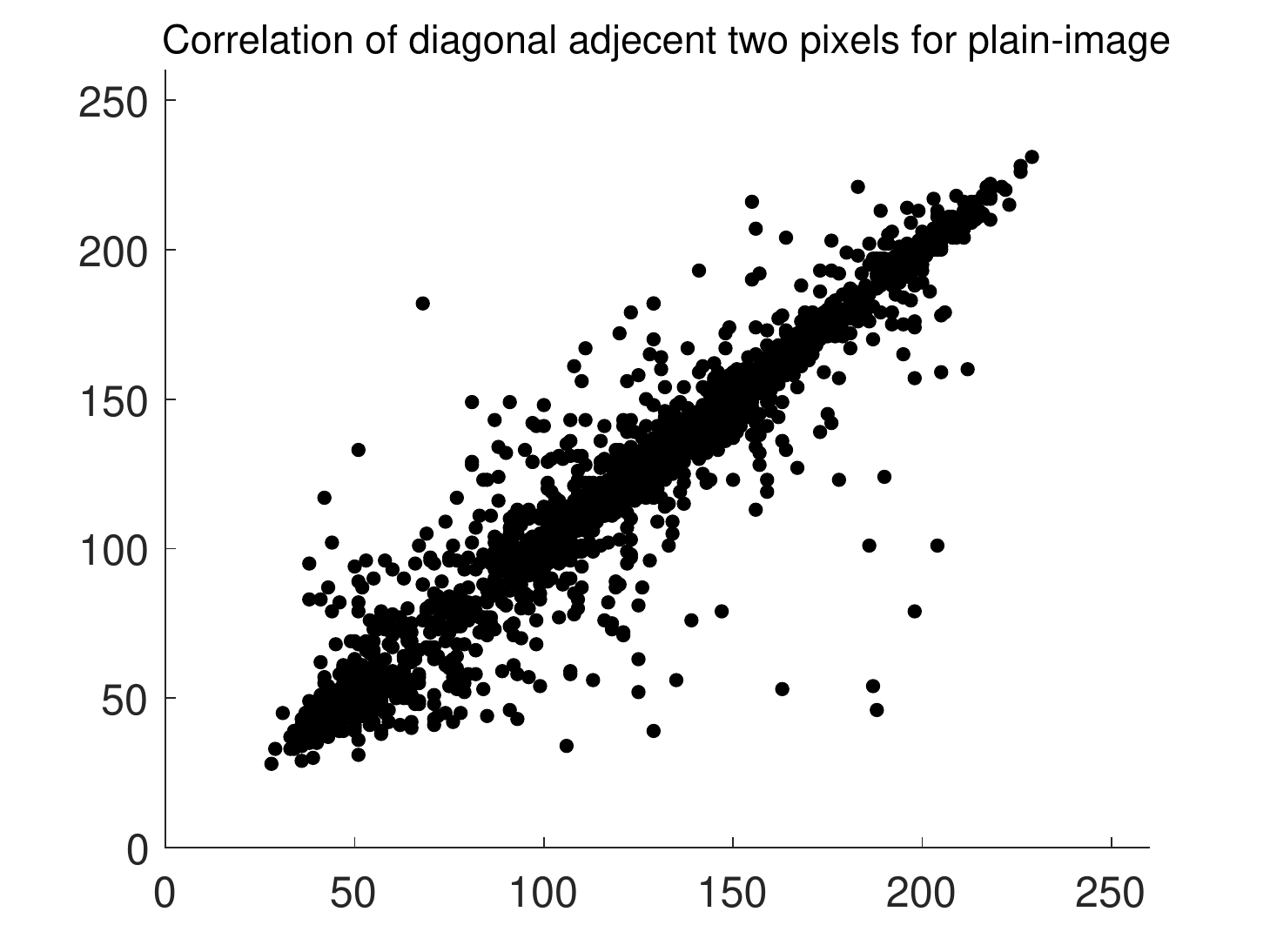}
	 \caption{}
	 \end{subfigure}
	 \begin{subfigure}[b]{0.25\textwidth}
	\centering
	 \includegraphics[scale=0.29]{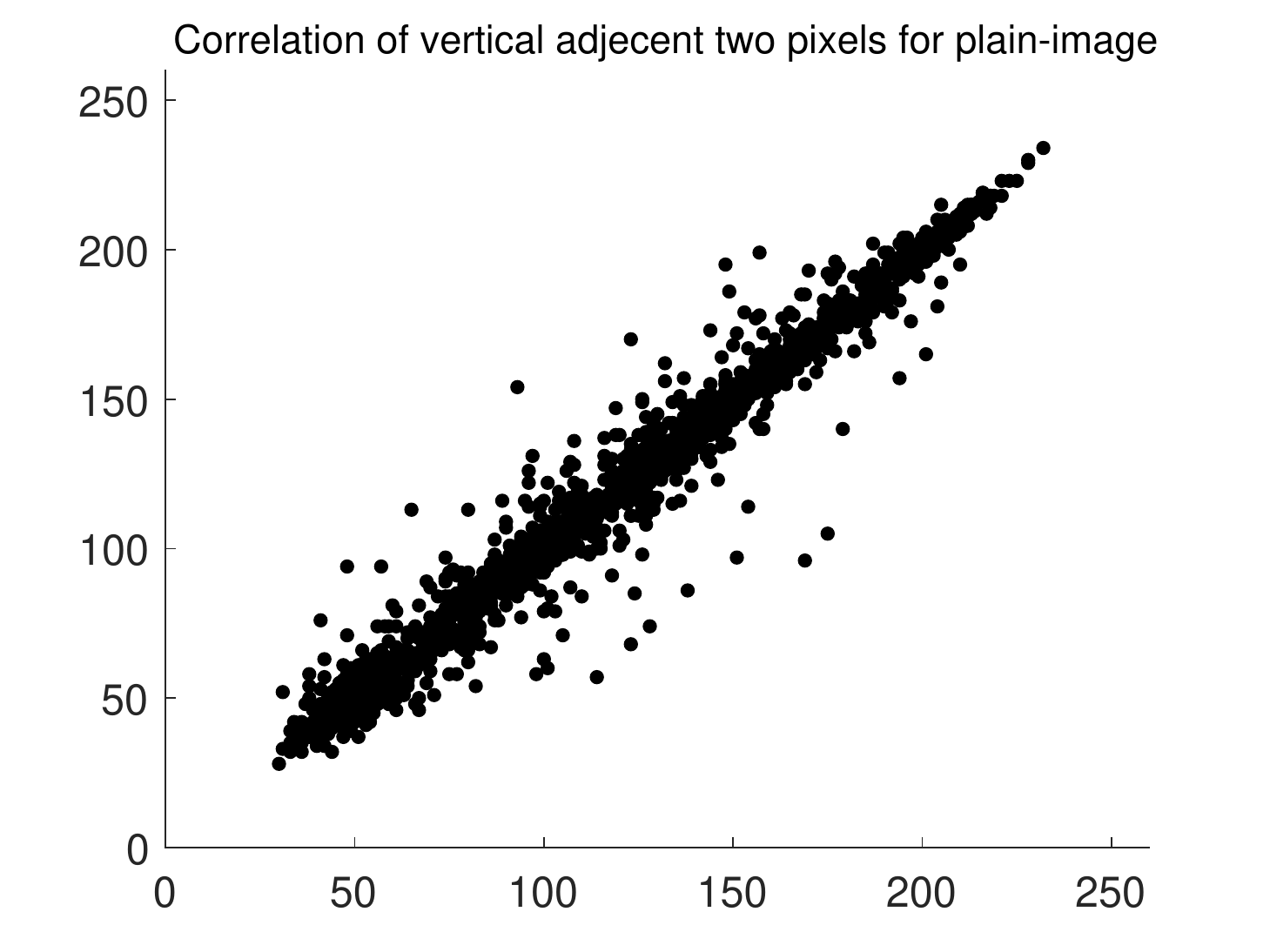}
	 \caption{ }
	 \end{subfigure}
\bigskip\\
\begin{subfigure}[b]{0.21\textwidth}
	\centering
	 \includegraphics[scale=0.21]{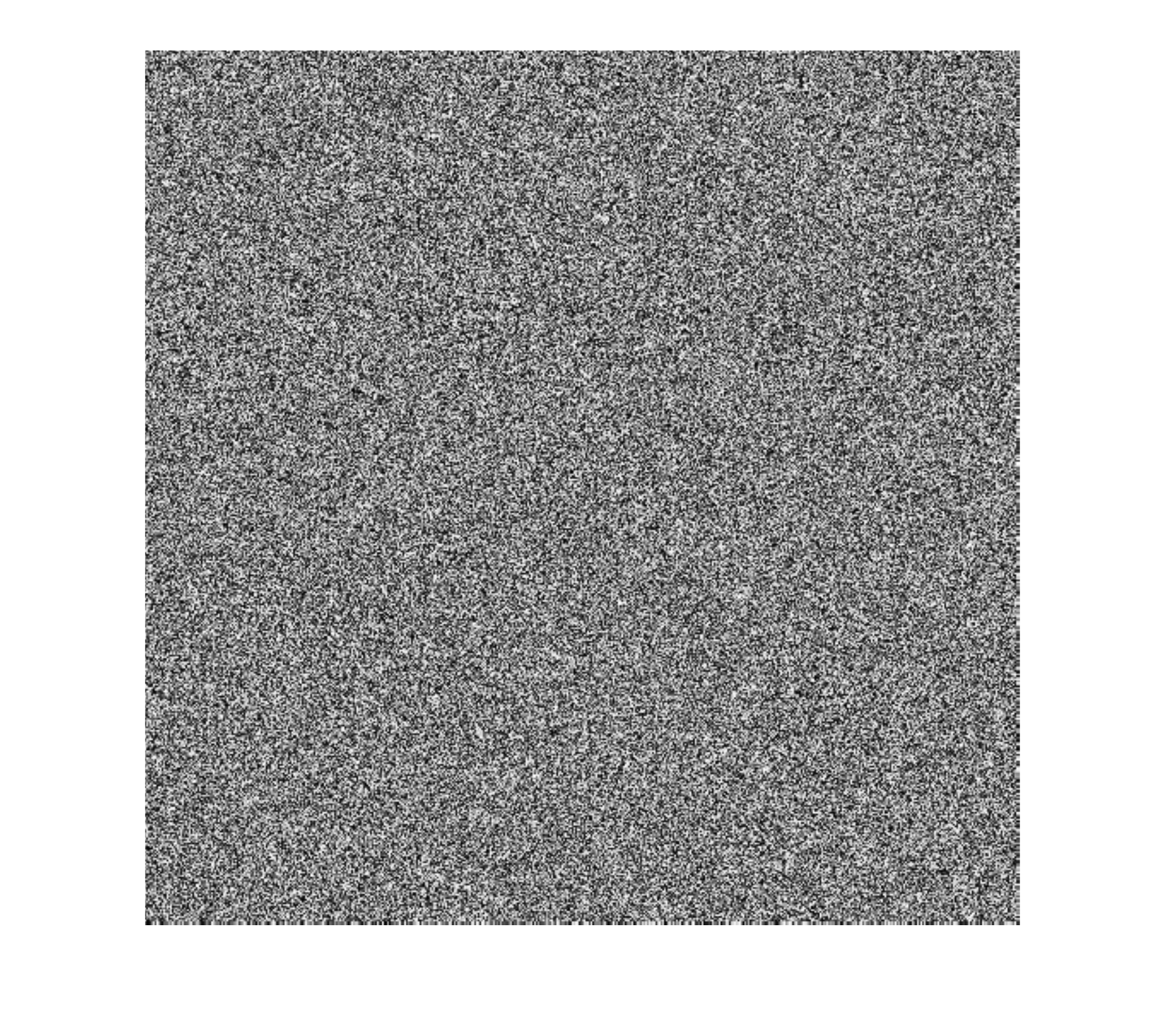}
	 \caption{}
	 \end{subfigure}
\begin{subfigure}[b]{0.25\textwidth}
	\centering
	 \includegraphics[scale=0.29]{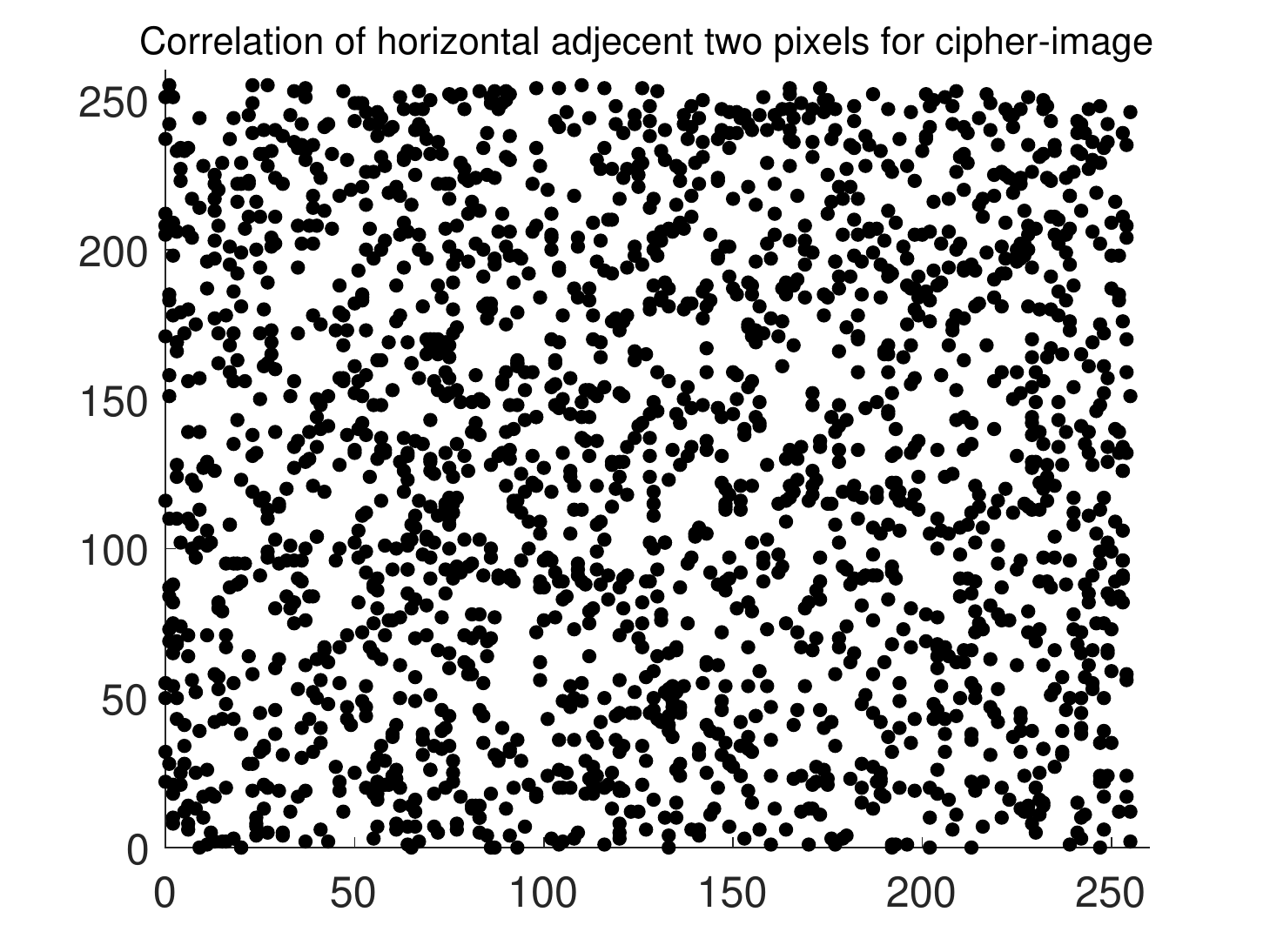}
	 \caption{ }
	 \end{subfigure}
	 \begin{subfigure}[b]{0.25\textwidth}
	\centering
	 \includegraphics[scale=0.29]{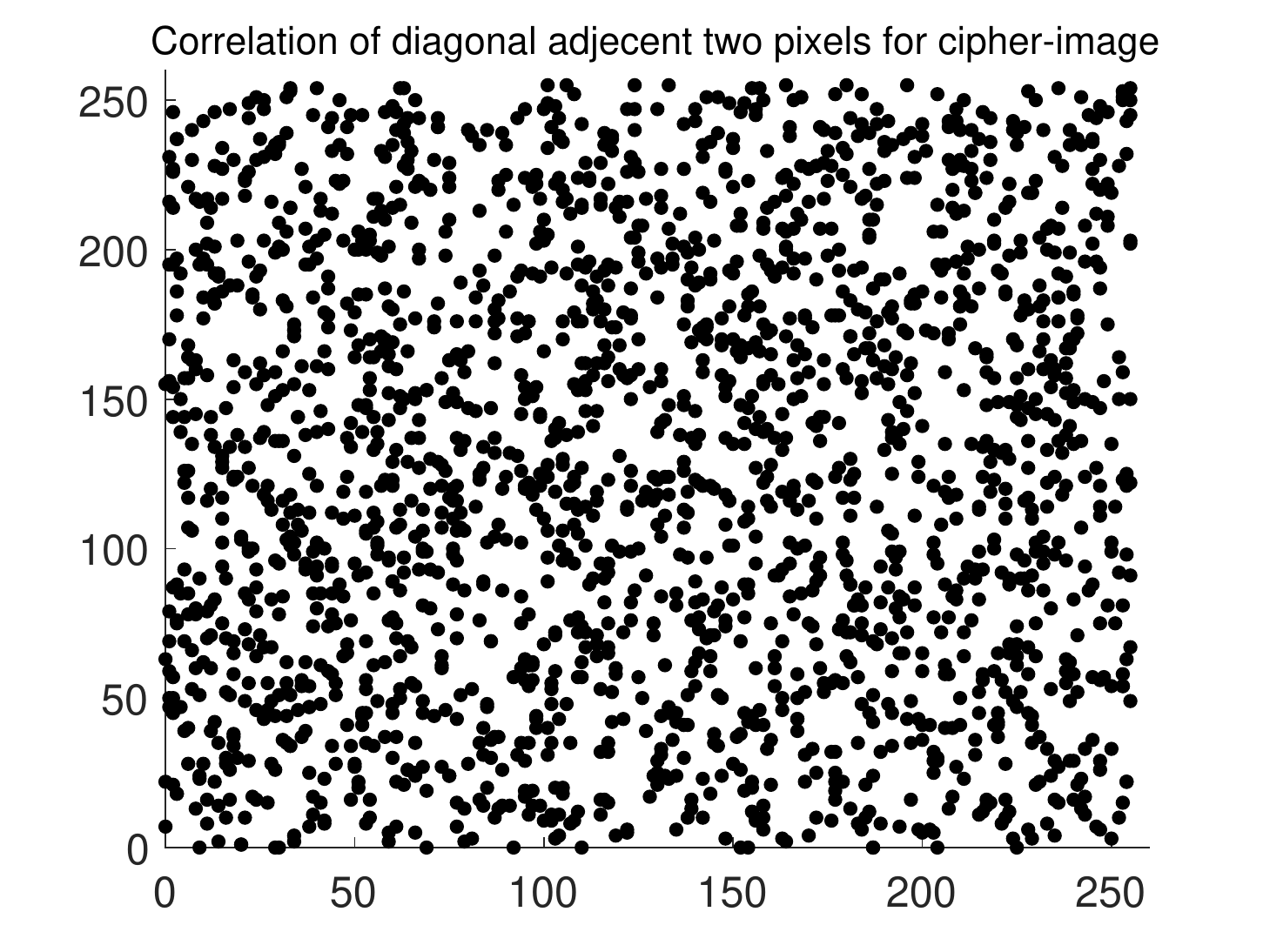}
	 \caption{}
	 \end{subfigure}
\begin{subfigure}[b]{0.25\textwidth}
	\centering
	 \includegraphics[scale=0.29]{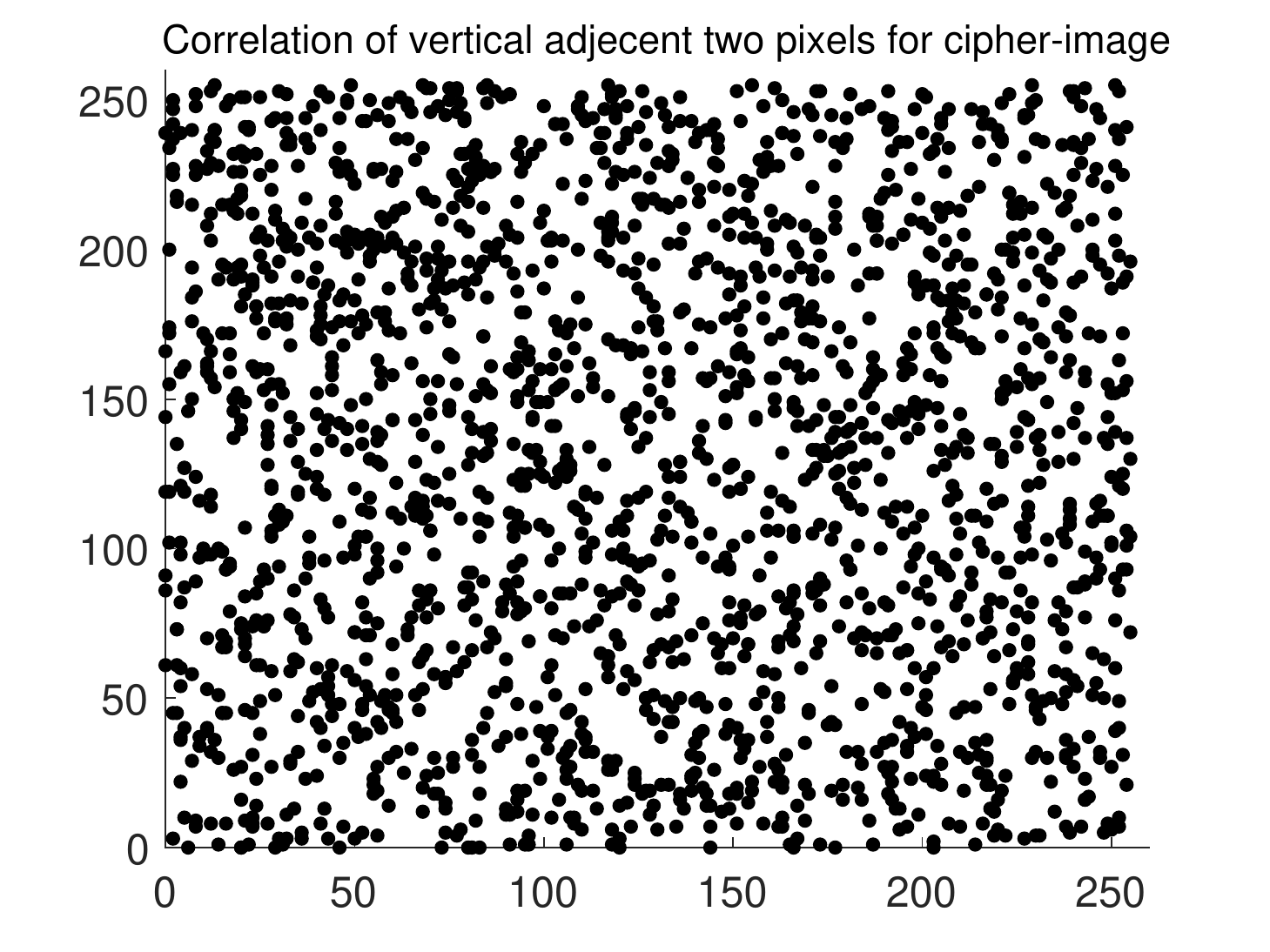}
	 \caption{ }
	 \end{subfigure}

\caption{{(\textbf{b})--(\textbf{d})~The distribution of pixels of the plane image in the horizontal, diagonal and vertical directions;
(\textbf{f})--(\textbf{h})~the distribution of pixels of the cipher image in the horizontal, diagonal and vertical~directions.}}
\label{fig:scatter}
\end{figure}
\end{enumerate}
\subsubsection{Differential Attack}
In differential attacks the opponents try to get the secret keys by studying the relation between the plain image and cipher image. Normally attackers encrypt two images by applying a small change to these images, then compare the properties of the corresponding cipher images. If a minor change in the original image can cause a significant change in the encrypted image, then the cryptosystem has a high security level. The two tests NPCR (number of pixels change rate) and UACI (unified average changing intensity) are usually used to describe the security level against differential attacks. For two plain images $P$ and $P^{'}$ different at only one pixel value, let  $C_{\rm P}$ and $C_{\rm P^{'}}$ be the cipher images of $P$ and $P^{'}$, respectively, then NPCR and UACI are calculated as:
\begin{align}
{\rm NPCR} &= \sum_{u=1}^{m} \sum_{v=1}^{n}\frac{\tau(u,v)} { m\times n},\label{NPCR}\\
{\rm UACI} &= \sum_{u=1}^{m} \sum_{v=1}^{n} \frac{|C_{\rm P}(u, v) - C_{\rm P^{'}}(u, v)|}{255\times m\times n}\label{UACI},
\end{align}
\textls[-15]{where $\tau(u,v)=0$ if $C_{\rm P}(u, v) = C_{\rm P^{'}}(u, v)$ and $\tau(u,v)=1,$ otherwise. The expected values of NPCR and UACI for 8-bit images are $0.996094$ and $0.334635$, respectively~\cite{Wu}. We applied the above two tests to each image of the database by randomly changing the pixel value of each image. The experimental results are shown in Figure~\ref{fig:NP_UA}, giving average values of NPCR and UACI of $0.9961$ and $0.3334$, respectively. It follows from the obtained results that our scheme is capable of resisting a differential attack.}
\begin{figure}[H]
\captionsetup[subfigure]{justification=centering}
\centering
\begin{subfigure}[b]{0.24\textwidth}
	\centering
	 \includegraphics[scale=0.275]{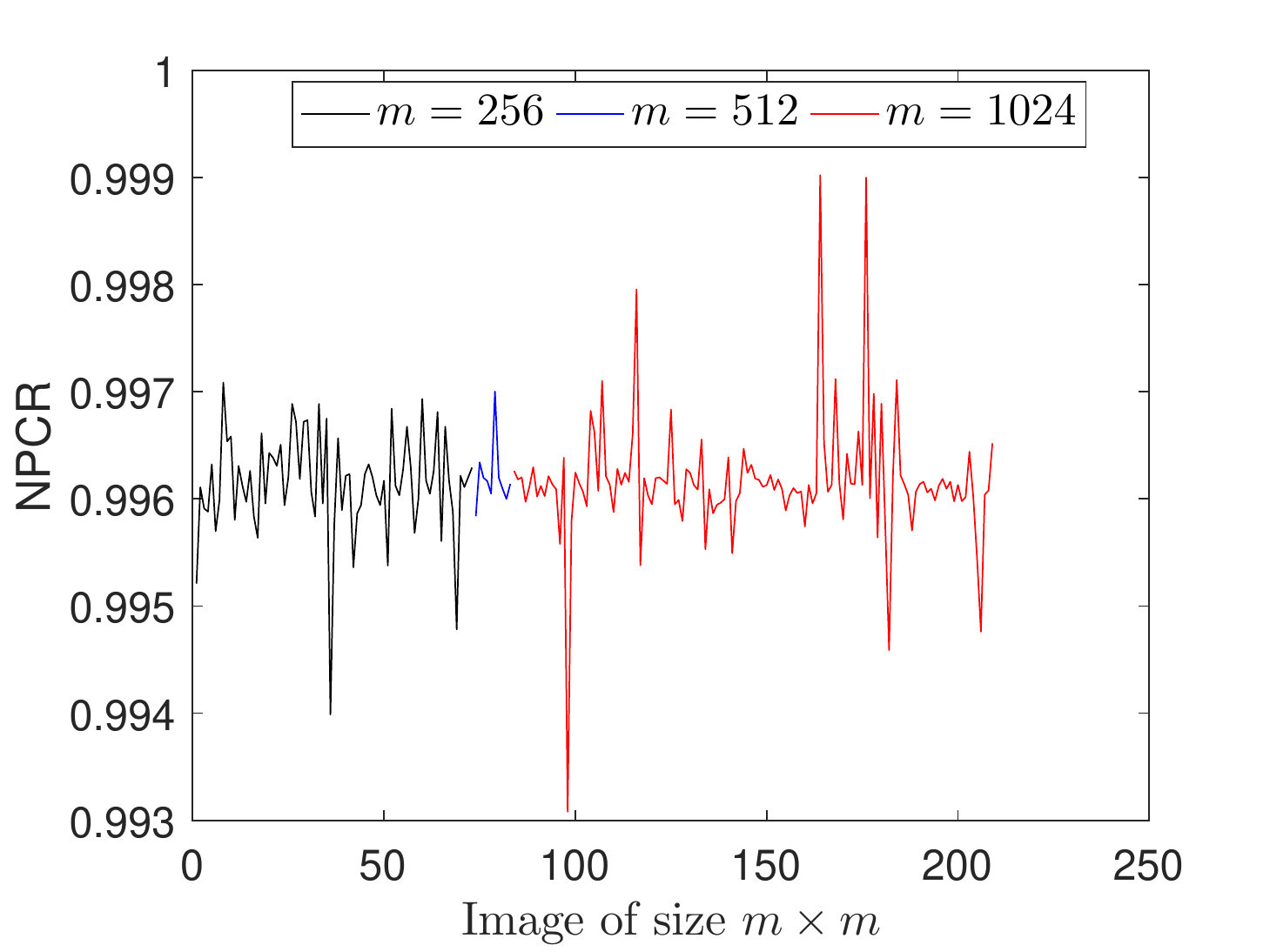}
	 \caption{ }
	 \end{subfigure}
\begin{subfigure}[b]{0.24\textwidth}
	\centering
	 \includegraphics[scale=0.275]{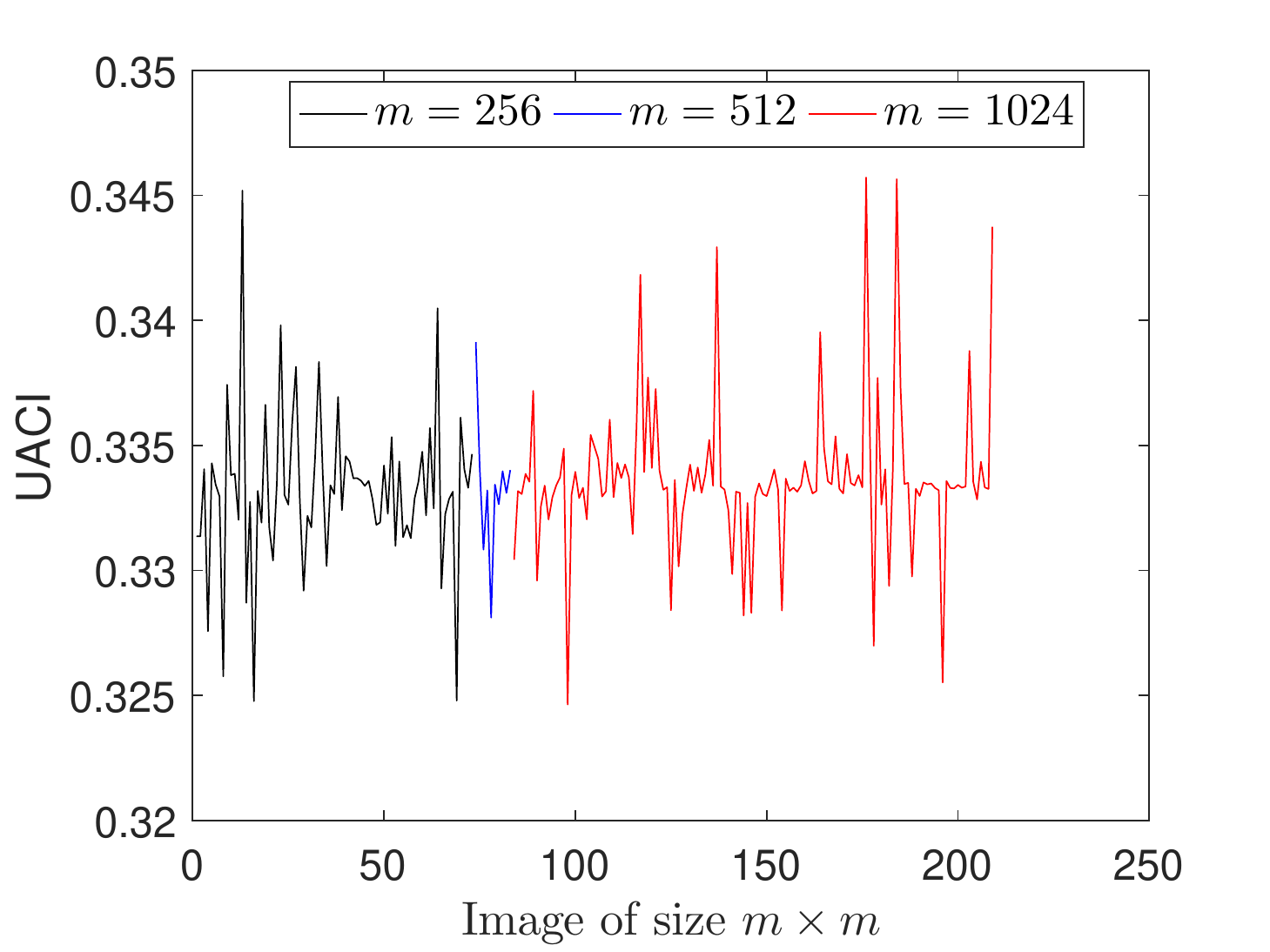}	
      \caption{ }
	 \end{subfigure}
	\centering
\begin{subfigure}[b]{0.24\textwidth}
	\centering
	 \includegraphics[scale=0.275]{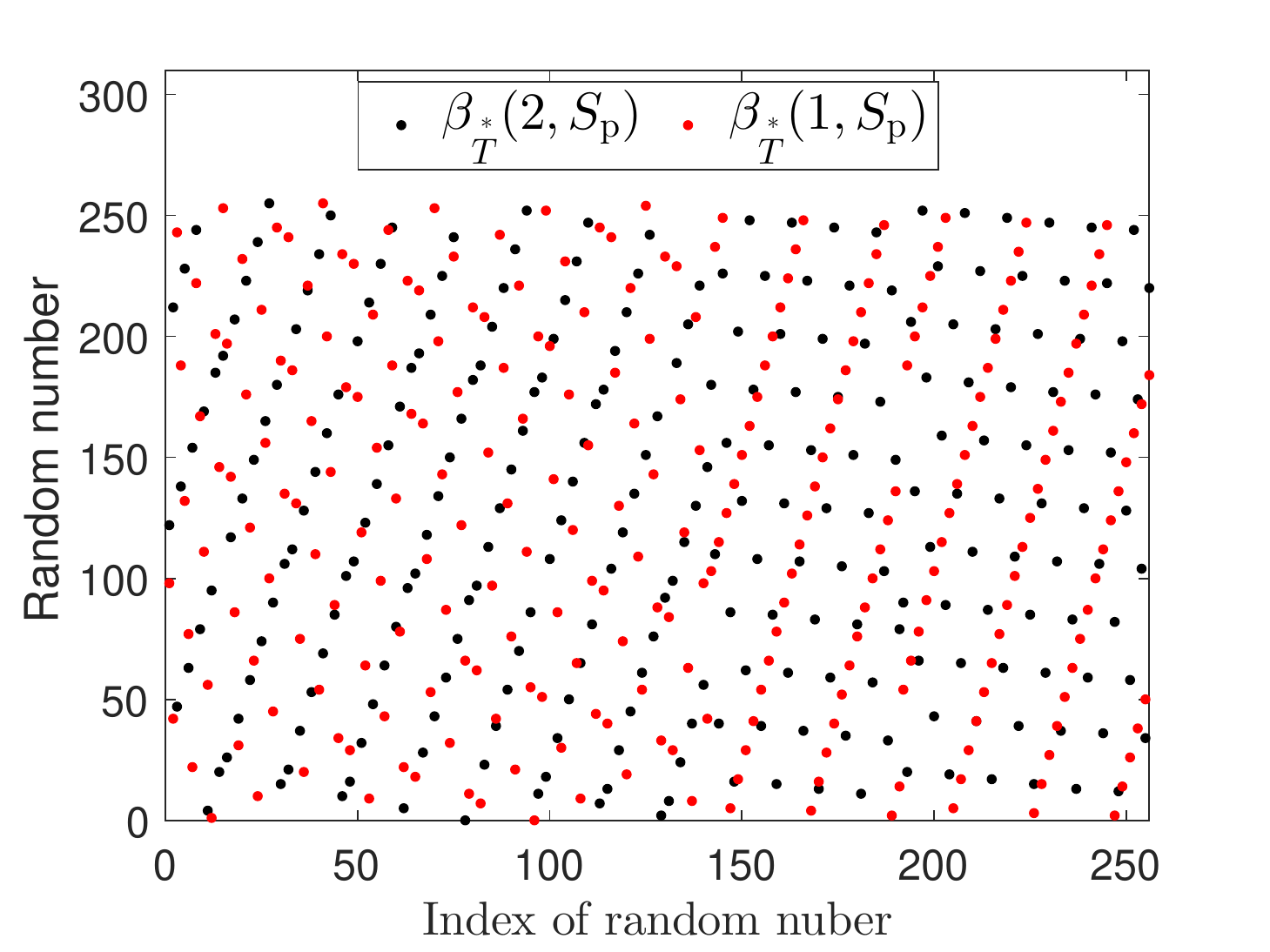}
	 \caption{ }
	 \end{subfigure}
\begin{subfigure}[b]{0.24\textwidth}
	\centering
	 \includegraphics[scale=0.275]{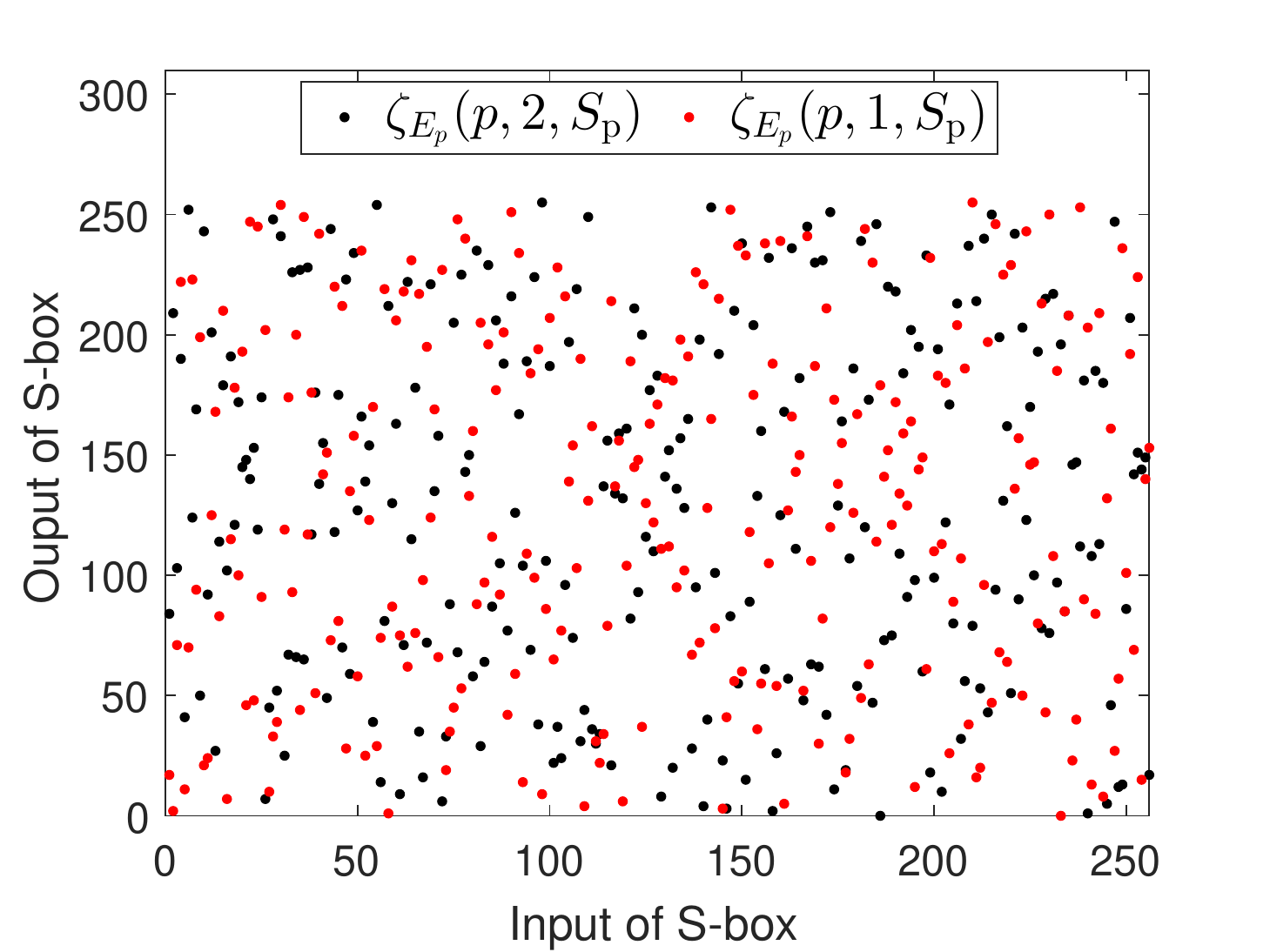}	
      \caption{ }
	 \end{subfigure}
\caption{(\textbf{a--b})~The NPCR and UACI results for each image in the USC-SIPI database; (\textbf{c})~First $256$ pseudo-random numbers and (\textbf{d})~two S-boxes  generated for Lena$_{512 \times 512}$ with a small change in an input key $t$.}
\label{fig:NP_UA}
\end{figure}
\subsubsection{Key Analysis}
\textls[-25]{For a secure cryptosystem it is essential to perform well against key attacks. A cryptosystem is highly secure against key attacks if it has key sensitivity and large key space and strongly opposes the known-plaintext/chosen-plaintext attack. The proposed scheme is analyzed against key attacks as follows.}
\begin{enumerate}
\item[(1)] Key sensitivity. Attackers usually use slightly different keys to encrypt a plain image and then compare the obtained cipher image with the original cipher image to get the actual keys. Thus, high key sensitivity is essential for higher security. That is, cipher images of a plain image generated by  two slightly different keys should be entirely different. The difference of the cipher images is quantified by Equations~(\ref{NPCR}) and~(\ref{UACI}). In experiments we encrypted the whole database by changing only one key, while other keys  remain unchanged. The key sensitivity results are shown in Table \ref{co}, where the average values of NPCR and UACI are $0.9960$ and $0.3341$, respectively, which specify the remarkable difference in the cipher images. Moreover, our cryptosytem is based on the pseudo-random numbers and S-boxes. The sensitivity of pseudo-random numbers sequences $\beta_{\stackrel{*}{T}}(2,S_{\rm P})$ and $\beta_{\stackrel{*}{T}}(1,S_{\rm P})$ and S-boxes $\zeta_{E_{p}}(p,2,S_{\rm P})$ and $\zeta_{E_{p}}(p,1,S_{\rm P})$ for Lena$_{512\times 512}$ is shown in Figure~\ref{fig:NP_UA}.
\begin{table}[H]
{\caption{Difference between two encrypted images when key $t=2$ is changed to $t=1$. NPCR: number of pixels change rate; UACI: unified average changing intensity.}
\label{co}}
\centering
\scalebox{0.94}
{\

\bgroup
\def\arraystretch{0.99}
{\begin{tabular}{ccccccccc}
\toprule
\textbf{Image}   & \textbf{NPCR(\%)} &\textbf{UACI(\%)}  & \textbf{Image} &\textbf{NPCR(\%)}  &\textbf{UACI(\%)} & \textbf{Image} &\textbf{NPCR(\%)}  &\textbf{UACI(\%)} \\ \midrule
Female& 99.62 &33.39  &House  &99.62&	33.23 &Couple  &99.56 &33.30  \\
Tree&99.59  &33.35  & Beans &99.64 &33.23 & Splash  &99.60 &33.97  \\ \bottomrule
\end{tabular}
\egroup
}}
\end{table}

\item[(2)] Key space. In order to resist a brute force attack, key space should be sufficiently large. For~any cryptosystem, key space represents the set of all possible keys required for the encryption process. Generally, the size of the key space should be greater than $2^{128}.$ In the present scheme the parameters $a_{1},b_{1},a_{2},b_{2},a_{3},\delta,L, S_{\rm P},t$ and $p$ are used as secret keys, and we store each of them in $28$ bits. Thus the key space of the proposed cryptosystem is $2^{280}$ which is larger than $2^{128}$ and hence capable to resist a brute force attack.
\item[(3)] Known-plaintext/chosen-plaintext attack. In a known-plaintext attack, the attacker has partial knowledge about the plain image and cipher image, and tries to break the cryptosystem, while in a chosen-plaintext attack the attacker encrypts an arbitrary image to get the encryption keys. An~all-white/black image is usually encrypted to test the performance of a scheme against these powerful attacks~\cite{CP,[ra]}. We analyzed our scheme by encrypting an all-white/black image of size $256 \times 256$. The results are shown in Figure~\ref{fig:whiteblack} and Table \ref{cmp}, revealing that the encrypted images are significantly randomized. Thus the proposed system is capable of preventing the above mentioned attacks.
    \begin{figure}[H]
\captionsetup[subfigure]{justification=centering}
\centering
\begin{subfigure}[b]{0.14\textwidth}
	\centering
	 \includegraphics[scale=0.152,frame]{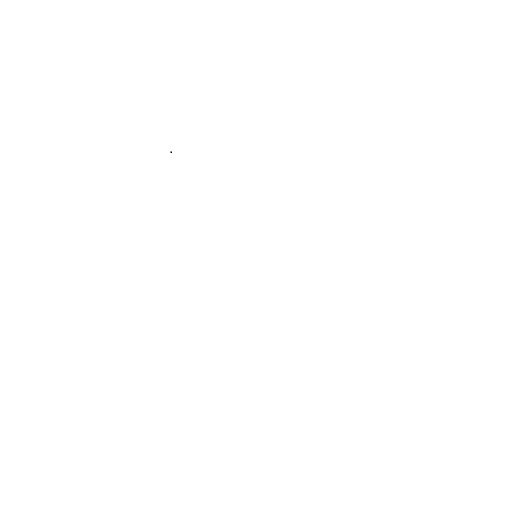}
	 \caption{ }
	 \end{subfigure}
\begin{subfigure}[b]{0.14\textwidth}
	\centering
	 \includegraphics[scale=0.155]{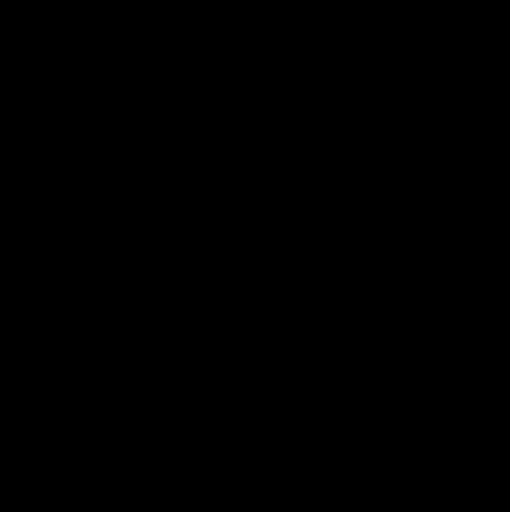}
	 \caption{}
	 \end{subfigure}
\begin{subfigure}[b]{0.14\textwidth}
	\centering
	 \includegraphics[scale=0.155]{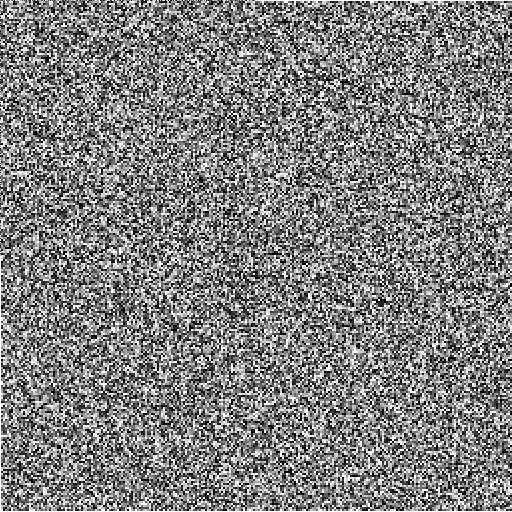}	
      \caption{ }
	 \end{subfigure}
\begin{subfigure}[b]{0.14\textwidth}
	\centering
	 \includegraphics[scale=0.155]{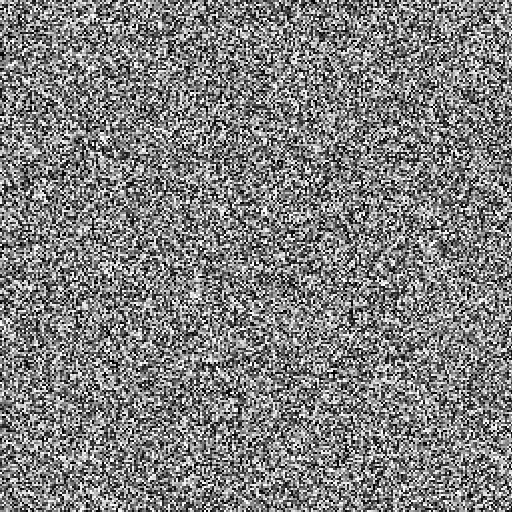}
	 \caption{ }
	 \end{subfigure}
	\begin{subfigure}[b]{0.17\textwidth}
	\centering
	 \includegraphics[scale=0.205]{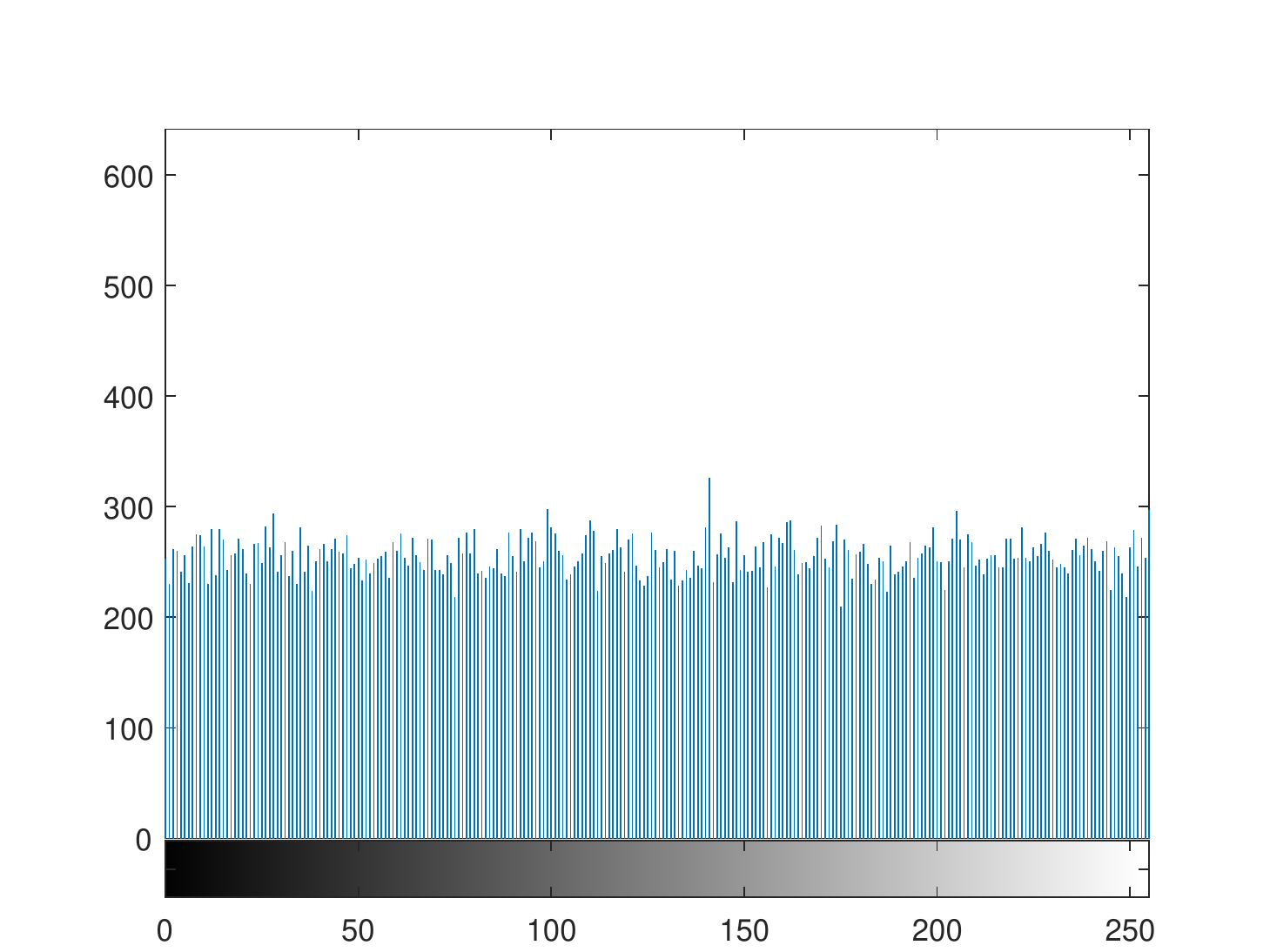}
	 \caption{}
	 \end{subfigure}
	 \begin{subfigure}[b]{0.20\textwidth}
	\centering
	 \includegraphics[scale=0.205]{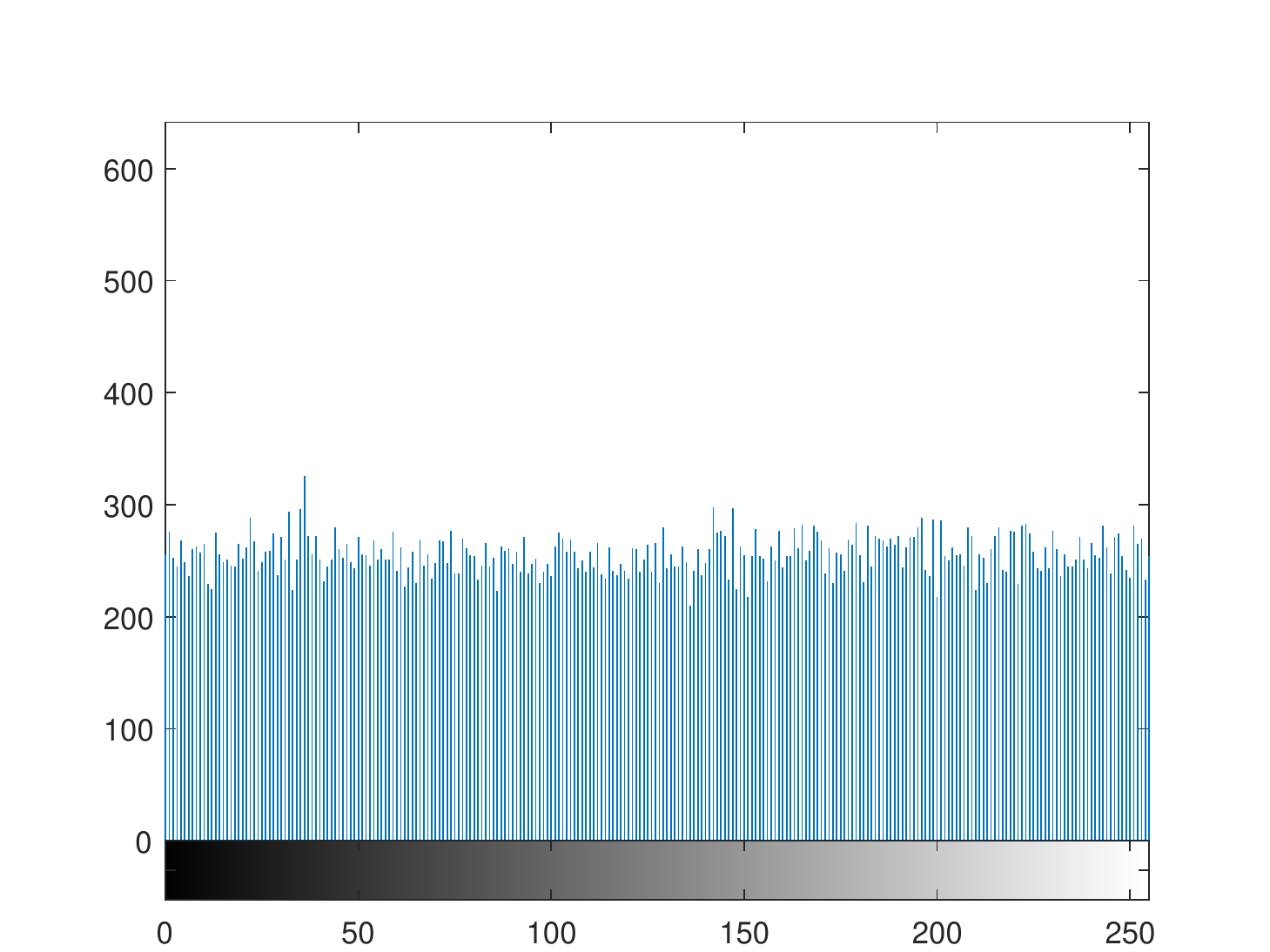}
	 \caption{}
	 \end{subfigure}
\caption{(\textbf{a}) All-white; (\textbf{b}) all-black; (\textbf{c})--(\textbf{d}) cipher images of (\textbf{a})--(\textbf{b}); (\textbf{e})--(\textbf{f}) histograms of (\textbf{c})--(\textbf{d}).} 
\label{fig:whiteblack}
\end{figure}\vspace{-12pt}

\begin{table}[H]
\parbox{13.5cm}
{\caption{Security analysis of all-white/black encrypted
images by the proposed encryption technique.}
\label{cmp}}
\centering
%
\begin{tabular}{ccccccc}
\toprule
\multirow{2}{*}{\textbf{Plain Image}}&\multirow{2}{*}{ \textbf{Entropy}}& \multicolumn{3}{c}{\textbf{Correlation of Plain Image}}&\multirow{2}{*}{\textbf{NPCR~(\%)}}&\multirow{2}{*}{\textbf{UACI~(\%)}} \\
\cmidrule{3-5} &  &\textbf{Hori.}& \textbf{Diag.}&\textbf{Ver.}&  &\\
\midrule
 All-white&7.9969&0.0027&0.0020&$-$0.0090&99.60	&33.45\\
 All-black&7.9969&$-$0.0080&0.0035&0.0057&99.62&33.41\\
\bottomrule
\end{tabular}
\end{table}

\end{enumerate}
\subsection{Comparison and Discussion}\label{Comparison}

%

Apart from security analyses, the proposed scheme is compared with some well-known image encryption techniques. The gray scale images of Lena$_{256\times256}$ and Lena$_{512\times512}$ are encrypted using the presented method, and experimental results are listed in Table~\ref{comparison}.
\begin{table}[H]
\caption{Comparison of the proposed encryption scheme with several existing cryptosystems for image Lena$_{m \times m}$, $m$ = 256,512.}
\centering
\resizebox{\columnwidth}{!}
{
\bgroup
\def\arraystretch{0.9}
{\begin{tabular}{l l l l l l l l c l}
\toprule
\multirow{2}{*}{\textbf{Size \boldmath{$m$}}} &\multirow{2}{*}{\textbf{Algorithm}} & \multirow{2}{*}{\textbf{Entropy}}  & \multicolumn{3}{c}{ \textbf{Correlation}} & \multirow{2}{*}{\textbf{NPCR (\%)}} & \multirow{2}{*}{\textbf{UACI(\%)}}& \boldmath{$\#$} & \textbf{Dynamic} \\ \cmidrule{4-6}
  & & &  \textbf{Hori.} & \textbf{Diag.} & \textbf{Ver.} &&  & \textbf{S-Boxes} & \textbf{S-Boxes} \\ \midrule
\multirow{7}{*}{$256$}&{\bf Ours} & \multicolumn{1}{c}{7.9974}  &
\multicolumn{1}{c}{0.0001
} & \multicolumn{1}{c}{$-$0.0007} & \multicolumn{1}{c}{
$-$0.0001} & \multicolumn{1}{c}{99.91} & \multicolumn{1}{c}{33.27}& \multicolumn{1}{c}{1} &
\multicolumn{1}{c}{Yes} \\
&Ref.~\cite{signal} & \multicolumn{1}{c}{7.9993}  &
\multicolumn{1}{c}{0.0012} & \multicolumn{1}{c}{0.0003} &
\multicolumn{1}{c}{0.0010} & \multicolumn{1}{c}{99.60} & \multicolumn{1}{c}{33.50}& \multicolumn{1}{c}{
1} & \multicolumn{1}{c}{Yes} \\
&Ref.~\cite{signal52} & \multicolumn{1}{c}{7.9973}  &
\multicolumn{1}{c}{-} & \multicolumn{1}{c}{-} & \multicolumn{1}{c}{
-} & \multicolumn{1}{c}{99.50} & \multicolumn{1}{c}{33.30}& \multicolumn{1}{c}{0} &
\multicolumn{1}{c}{-} \\
&{Ref.~\cite{signal46}} & \multicolumn{1}{c}{7.9046}  &
\multicolumn{1}{c}{0.0164} & \multicolumn{1}{c}{$-$0.0098} & \multicolumn{1}{c}{
0.0324} & \multicolumn{1}{c}{98.92} & \multicolumn{1}{c}{32.79}& \multicolumn{1}{c}{{\textgreater}1{\textless}50} & \multicolumn{1}{c}{Yes} \\
&Ref.~\cite{[ra23]} & \multicolumn{1}{c}{7.9963}  &
\multicolumn{1}{c}{$-$0.0048} & \multicolumn{1}{c}{$-$0.0045} & \multicolumn{1}{c}{
$-$0.0112} & \multicolumn{1}{c}{99.62} & \multicolumn{1}{c}{33.70}& \multicolumn{1}{c}{8} &
\multicolumn{1}{c}{Yes} \\
&Ref.~\cite{[rb]} & \multicolumn{1}{c}{7.9912}  &\multicolumn{1}{c}{$-$0.0001} & \multicolumn{1}{c}{0.0091} & \multicolumn{1}{c}{
0.0089} & \multicolumn{1}{c}{100} & \multicolumn{1}{c}{33.47}& \multicolumn{1}{c}{0} &
\multicolumn{1}{c}{-} \\
&Ref.~\cite{Wan} &\multicolumn{1}{c}{7.9974}&\multicolumn{1}{c}{0.0020}&\multicolumn{1}{c}{0.0020}&
\multicolumn{1}{c}{0.0105}&\multicolumn{1}{c}{99.59}&\multicolumn{1}{c}{33.52}&
\multicolumn{1}{c}{0}&\multicolumn{1}{c}{-}\\ \midrule
\multirow{3}{*}{$512$}&{\bf Ours} & \multicolumn{1}{c}{7.9993} &
\multicolumn{1}{c}{0.0001} & \multicolumn{1}{c}{0.0042} & \multicolumn{1}{c}{
0.0021} & \multicolumn{1}{c}{99.61} & \multicolumn{1}{c}{33.36} & \multicolumn{1}{c}{1} &
\multicolumn{1}{c}{Yes} \\
&Ref.~\cite{Cheng} & \multicolumn{1}{c}{7.9992}  &
\multicolumn{1}{c}{0.0075} & \multicolumn{1}{c}{0.0016} & \multicolumn{1}{c}{0.0057} &
\multicolumn{1}{c}{99.61}& \multicolumn{1}{c}{33.38} & \multicolumn{1}{c}{1} & \multicolumn{1}{c}{No}
\\
&Ref.~\cite{[ra]} & \multicolumn{1}{c}{7.9993}  &
\multicolumn{1}{c}{$-$0.0004} & \multicolumn{1}{c}{0.0001} &
\multicolumn{1}{c}{$-$0.0018} & \multicolumn{1}{c}{99.60} & \multicolumn{1}{c}{33.48}& \multicolumn{1}{c}{
1} & \multicolumn{1}{c}{No} \\
\midrule
\multirow{3}{*}{-}&Ref. \cite{signal43} & \multicolumn{1}{c}{7.9970} &
\multicolumn{1}{c}{$-$0.0029} & \multicolumn{1}{c}{0.0135} & \multicolumn{1}{c}{
0.0126} & \multicolumn{1}{c}{99.60}& \multicolumn{1}{c}{33.48}  & \multicolumn{1}{c}{0} &
\multicolumn{1}{c}{-} \\ 
&Ref.~\cite{signal51} & \multicolumn{1}{c}{7.9994}  &
\multicolumn{1}{c}{0.0018} & \multicolumn{1}{c}{$-$0.0012} & \multicolumn{1}{c}{
0.0011} & \multicolumn{1}{c}{99.62} & \multicolumn{1}{c}{33.44}& \multicolumn{1}{c}{{\textgreater}1} &
\multicolumn{1}{c}{Yes} \\
&Ref.~\cite{signal45} & \multicolumn{1}{c}{7.9993}  &
\multicolumn{1}{c}{0.0032} & \multicolumn{1}{c}{0.0011} &
\multicolumn{1}{c}{$-$0.0002} & \multicolumn{1}{c}{99.60}& \multicolumn{1}{c}{33.47} & \multicolumn{1}{c}{{\textgreater}1} & \multicolumn{1}{c}{Yes} \\
 \bottomrule
\end{tabular}}
\egroup}
\label{comparison}%
\end{table}%
It is deduced that our scheme generates cipher images with comparable security. Furthermore, we remark that the scheme in~\cite{[ra]} generates pseudo-random numbers using group law on EC, while the proposed method generates pseudo-random numbers by constructing triads using auxiliary parameters of elliptic surfaces. Group law consists of many operations, which makes the pseudo-random number generation process slower than the one we present here.
The scheme in~\cite{[ra23]} decomposes an image to eight blocks and uses dynamic S-boxes for encryption purposes. The computation of multiple S-boxes takes more time than computing only one S-box. Similarly the techniques in~\cite{signal45,signal46} use a set of S-boxes and encrypt an image in blocks, while our newly developed scheme encrypts the whole image using only one dynamic S-box. Thus, our scheme is faster than the schemes in~\cite{signal45,signal46}. The security system in~\cite{signal43} uses a chaotic system to encrypt blocks of an image. The~results in Table~\ref{comparison} reveal that our proposed system is cryptographically stronger than the scheme in~\cite{signal43}. The algorithms in~\cite{[rb],signal52} combine  chaotic systems and different ECs to encrypt images. It~follows from Table~\ref{comparison} that the security level of our scheme is comparable to that of the schemes in~\cite{[rb],signal52}. The technique in~\cite{Wan} uses double chaos along with DNA coding to get good results, as shown in Table~\ref{comparison}, but the results obtained by the new scheme are better than that of~\cite{Wan}.  
Similarly the technique in~\cite{signal} encrypts images using ECs but  does not guarantee an S-box for each set of input parameters, thus making our scheme faster and more robust than the scheme developed in~\cite{signal}.

Furthermore, the following facts put our scheme in a favorable position:
\begin{itemize}[align=parleft,leftmargin=*,labelsep=7.5mm]
\item[(i)] Our scheme uses a dynamic S-box for each input image while the S-box used in ~\cite{[ra]} is a static one, which is vulnerable~\cite{Rosenthal} and less secure than a dynamic one~\cite{Kazlauskas}.
\item[(ii)]  The presented scheme guarantees an S-box for each image, which is not the case in ~\cite{signal}.
\item[(iii)] \textls[-15]{To get random numbers, the described scheme  generates triads for all images of the same size, while in~\cite{signal} the computation of an EC for each input image is necessary, which is time consuming.}
\item[(iv)] The scheme in~\cite{[ra23]} uses eight dynamic S-boxes for a plain image, while the current scheme uses only one dynamic S-box for each image to get the desired cryptographic security.
\end{itemize}
\section{Conclusions}\label{Conclusion}
 \textls[-15]{An image encryption scheme based on quasi-resonant triads and MECs was introduced. The~proposed technique constructs triads to generate pseudo-random numbers and computes an MEC to construct an S-box for each input image. The pseudo-random numbers and S-box are then used for altering and scrambling the pixels of the plain image, respectively.  As for the advantages of our proposed method, firstly triads are based on auxiliary parameters of elliptic surfaces, and thus pseudo-random numbers and S-boxes generated by our method are highly sensitive to the plain image, which prevents adversaries from initiating any successful attack. Secondly, generation of triads using auxiliary parameters of elliptic surfaces consumes less time than computing points on ECs {(we find a 4x speed increase for a range of image resolutions $m \in [128, 512]$)}, which makes the new encryption system relatively faster.  Thirdly, our algorithm generates the cipher images with an appropriate security level.}

In summary, all of the above analyses imply that the presented scheme is able to resist all attacks. It has high encryption efficiency and less time complexity than some of the existing techniques. In the future, the current scheme will be further optimized by means of new ideas to construct the S-boxes using the constructed triads, so that  we will not need to compute an MEC for each input image.
\vspace{6pt}

\authorcontributions{All authors contributed equally to this work.}

\funding{This research is funded through the HEC project NRPU-7433.}

\acknowledgments{We thank Gene Kopp for useful comments and suggestions.}

\conflictsofinterest{The authors declare no conflict of interest. The funding sponsors had no role in the design of the study; in the collection, analyses, or interpretation of data; in the writing of the manuscript, and in the decision to publish the results.}

\abbreviations{The following abbreviations are used in this manuscript:\\

\noindent
\begin{tabular}{@{}ll}
MEC & Mordell elliptic curve\\
S-box & Substitution box\\
EC & Elliptic curves\\
\end{tabular}}

\appendixtitles{no} 


\reftitle{References}

\end{document}